\documentclass[english,final]{IEEEtran}
\usepackage[T1]{fontenc}
\usepackage[latin9]{inputenc}
\usepackage{xcolor}
\usepackage{pdfcolmk}
\usepackage{bm}
\usepackage{amsmath}
\usepackage{amsthm}
\usepackage{amssymb}
\usepackage{graphicx}
\PassOptionsToPackage{normalem}{ulem}
\usepackage{ulem}

\makeatletter

\providecolor{lyxadded}{rgb}{0,0,1}
\providecolor{lyxdeleted}{rgb}{1,0,0}
\DeclareRobustCommand{\lyxdeleted}[3]{{\color{lyxdeleted}\lyxsout{#3}}}
\DeclareRobustCommand{\lyxsout}[1]{\ifx\\#1\else\sout{#1}\fi}

\theoremstyle{plain}
\newtheorem{thm}{\protect\theoremname}
\theoremstyle{plain}
\newtheorem{prop}{\protect\propositionname}
\theoremstyle{plain}
\newtheorem{cor}{\protect\corollaryname}

\usepackage{amsmath}
\usepackage{amssymb}

\usepackage[level=3]{wgroup_message}  % set level to 0, to hide all notes
\usepackage{graphicx,psfrag,cite,subfigure}
\author{
\IEEEauthorblockN{Hao~Sun and Junting~Chen}

\IEEEauthorblockA{}
\thanks{H. Sun and J. Chen are with the School of Science and Engineering and Future Network Intelligence Institute (FNii), The Chinese University of Hong Kong, Shenzhen, Guangdong 518172, China.}
}


\usepackage[acronym]{glossaries}
\newcommand{\newac}{\newacronym}
\newcommand{\ac}{\gls}
\newcommand{\Ac}{\Gls}

\newac{psd2}{PSD}{power spectral density}
\newac{slf}{SLF}{spatial loss function}
\newac{btd}{BTD}{block-term tensor decomposition}
\newac{tps}{TPS}{thin plate spline}
\newac{nmse}{NMSE}{normalized mean squared error}
\makeglossaries

\newac{speb}{SPEB}{square position error bound}
\newac[plural=EFIMs,firstplural=Fisher information matrices (EFIMs)]{efim}{EFIM}{Fisher information matrix}
\newac{ne}{NE}{Nash equilibrium}
\newac{mse}{MSE}{mean squared error}
\newac{toa}{TOA}{time-of-arrival}
\newac{snr}{SNR}{signal-to-noise ratio}
\newac{lan}{LAN}{local area network}
\newac{psd}{PSD}{positive semidefinite}
\newac{pd}{PD}{positive definite}
\newac{wrt}{w.r.t.}{with respect to}
\newac{lhs}{L.H.S.}{left hand side}
\newac{wp1}{w.p.1}{with probability 1}
\newac{kkt}{KKT}{Karush-Kuhn-Tucker}
\newac{wlog}{w.l.o.g.}{without loss of generality}
\newac{mle}{MLE}{maximum likelihood estimation}
\newac{gps}{GPS}{global positioning system}
\newac{rssi}{RSSI}{received signal strength indicator}
\newac{mimo}{MIMO}{multiple-input multiple-output}
\newac{csi}{CSI}{channel state information}
\newac{fdd}{FDD}{frequency division duplexing}
\newac{ms}{MS}{mobile station}
\newac{bs}{BS}{base station}
\newac{d2d}{D2D}{device-to-device}
\newac{slnr}{SLNR}{signal-to-interference-leakage-and-noise-ratio}
\newac{ula}{ULA}{uniform linear antenna array}
\newac{pas}{PAS}{power angular spectrum}
\newac{mmse}{MMSE}{minimum mean square error}
\newac{zf}{ZF}{zero-forcing}
\newac{rzf}{RZF}{regularized zero-forcing}
\newac{as}{AS}{angular spread}
\newac{aod}{AOD}{angle of departure}
\newac{iid}{i.i.d.}{independent and identically distributed} 
\newac{sinr}{SINR}{signal-to-interference-and-noise ratio}
\newac{tdd}{TDD}{time-division duplex}
\newac{rvq}{RVQ}{random vector quantization}
\newac{rhs}{R.H.S.}{right hand side}
\newac{mrc}{MRC}{maximum ratio combining}
\newac{cdf}{CDF}{cumulative distribution function}
\newac{a.s.}{a.s.}{almost surely}
\newac{los}{LOS}{line-of-sight}
\newac{jsdm}{JSDM}{joint spatial division and multiplexing}
\newac{map}{MAP}{maximum a posteriori}
\newac{klt}{KLT}{Karhunen-Lo\`eve Transform}
\newac{lbe}{LBE}{link bargaining equilibrium}
\newac{se}{SE}{Stackelberg equilibrium}
\newac{uav}{UAV}{unmanned aerial vehicle}
\newac{nlos}{NLOS}{non-line-of-sight}
\newac{pdf}{PDF}{probability density function}
\newac{em}{EM}{expectation-maximization}
\newac{knn}{KNN}{$k$-nearest neighbor}
\newac{svd}{SVD}{singular value decomposition}
\newac{nmf}{NMF}{non-negative matrix factorization}
\newac{umf}{UMF}{unimodality-constrained matrix factorization}
\newac{rmse}{RMSE}{rooted mean squared error}
\newac{olos}{OLOS}{obstructed line-of-sight}
\newac{mmw}{mmW}{millimeter wave}
\newac{ber}{BER}{bit error rate}
\newac{rss}{RSS}{received signal strength}
\newac{lp}{LP}{linear program}
\newac{ufw}{U-FW}{unimodal Frank-Wolfe}
\newac{utf}{UTF}{unimodality-constrained tensor factorization}
\newac{fw}{FW}{Frank-Wolfe}
\newac{iot}{IoT}{Internet-of-Things}
\newac{mae}{MAE}{mean absolute error}
\newac{crb}{CRB}{Cram\'er-Rao bound}
\newac{aoa}{AoA}{angle of arrival}
\newac{wcl}{WCL}{weighted centroid localization}

\setkeys{Gin}{width=1.0\columnwidth}

\renewcommand{\lyxdeleted}[3]{{\color{lyxdeleted}{}}}
\makeatother

\usepackage{babel}
\providecommand{\corollaryname}{Corollary}
\providecommand{\propositionname}{Proposition}
\providecommand{\theoremname}{Theorem}

\begin{document}
\title{Integrated Interpolation and Block-term Tensor Decomposition for Spectrum
Map Construction}
\maketitle
\begin{abstract}
This paper addresses the challenge of reconstructing a 3D power spectrum
map from sparse, scattered, and incomplete spectrum measurements.
It proposes an integrated approach combining interpolation and block-term
tensor decomposition (BTD). This approach leverages an interpolation
model with the BTD structure to exploit the spatial correlation of
power spectrum maps. Additionally, nuclear norm regularization is
incorporated to effectively capture the low-rank characteristics.
To implement this approach, a novel algorithm that combines alternating
regression with singular value thresholding is developed. Analytical
justification for the enhancement provided by the BTD structure in
interpolating power spectrum maps is provided, yielding several important
theoretical insights. The analysis explores the impact of the spectrum
on the error in the proposed method and compares it to conventional
local polynomial interpolation. Extensive numerical results demonstrate
that the proposed method outperforms state-of-the-art methods in terms
of signal source separation and power spectrum map construction, and
remains stable under off-grid measurements and inhomogeneous measurement
topologies.
\end{abstract}

\begin{IEEEkeywords}
Integrated, interpolation, block-term tensor decomposition, alternating
minimization, sparse observations, power spectrum map.
\end{IEEEkeywords}

\section{Introduction}

Spectrum maps, or more specifically, power spectrum maps, enable various
applications in wireless signal processing and communications, such
as signal propagation modeling \cite{RomKim:J22,GesEsrCheGan:J23},
source localization \cite{GaoWuYin:J23,CheMit:J17}, wireless power
transfer \cite{MoxHuaXuj:J19}, and radio resource management \cite{SheWanDin:J22}.
It has been a challenging problem in constructing power spectrum maps.
Firstly, measurement data is only available in a limited locations
or along a few routes, but the spatial pattern of the power spectrum
map can be very complex due to possible signal reflection and attenuation
from the propagation environment. Second, the location and the power
spectrum of the signal source can be time-varying, and therefore,
the power spectrum map should be constructed within a limited time
based on limited measurements.

There has been active research on power spectrum map construction.
Traditional interpolation-based approaches construct each map point
as a linear combination of nearby measurements, where the weights
can be found using different methods including Kriging \cite{SatSutIna:J21,SatKoyFuj:J17},
local polynomial regression \cite{Fan:b96}, and kernel-based methods
\cite{ZhaWan:J22,HamBefBal:C17}. Compressive sensing inspired approaches
exploit the fact that the matrix representation of a power spectrum
map has a \emph{low-rank} structure, and thus, dictionary learning
\cite{KimGia:C13} and matrix or tensor completion \cite{MigMarDon:J11,SunChe:C21,SchCavSta:C19,SheWanDin:J22}
can be applied to interpolate or extrapolate the power spectrum information
at locations without measurements. Recent work \cite{SheWanDin:J22}
proposed an orthogonal matching pursuit using a tensor model for 3D
spectrum mapping. The approach in \cite{SchCavSta:C19} reconstructed
the power spectrum map through minimizing the tensor rank while also
enforcing smoothness. Deep learning-based approaches \cite{TegRom:J21,HuHuaChen:J23}
treated the power spectrum map as one or multiple layers of 2D images,
and neural networks were trained to memorize the common pattern of
these images based on a huge amount of training data.

Many of these approaches were originally developed by representing
the power spectrum map as one or multiple 2D data layers, for example,
each layer representing a 2D power spectrum map for a specific frequency
band. Although most existing approaches can be directly extended to
the case in a higher dimension, their construction performance substantially
degrades due to the curse of dimensionality unless a special structure
is utilized. To address this issue, recent studies proposed a \ac{btd}
model to exploit the spatial correlation of the power spectrum map
at different frequency bands \cite{ZhaFuWan:J20,ShrFuHong:J22}. The
\ac{btd} model captures the property that if a signal from a particular
source in a certain frequency band is blocked at a specific location,
it is likely that signals from the same source at different frequency
bands will also be blocked at that location. In \cite{ZhaFuWan:J20},
a power spectrum map was reconstructed using a \ac{btd} model where
power spectrum and spatial loss field were constructed separately.
The method in \cite{ShrFuHong:J22} applies a deep neural network
to the tensor decomposition model to learn the intricate underlying
structures.

However, tensor decomposition for power spectrum map construction
has several limitations. First, the assumption of low-rank properties
in matrix or tensor models relies on on-grid measurements, a condition
rarely met in reality due to physical constraints on sensor placements.
Off-grid measurements induce discretization errors that likely destroy
the low-rank property, and consequently, degrade the performance of
matrix or tensor completion. Second, while one can reduce the grid
size to reduce the discretization error, a small grid size will translate
to a large matrix dimension, and this leads to the identifiability
issue as the matrix or tensor can be too sparse to complete. To address
the issue of off-grid measurements, recent works proposed to joint
interpolation with matrix or tensor completion \cite{SunChe:J22,CheWanZha:J23}.
The work \cite{SunChe:J22} applied adaptive interpolation to estimate
the values at the grid points, and then, performed uncertainty-aware
matrix completion according to the estimated interpolation errors.
In \cite{CheWanZha:J23}, \ac{tps} interpolation was used to construct
each layer of the tensor, followed by \ac{btd} to construct the power
spectrum map.

However, existing approaches \cite{SunChe:J22,CheWanZha:J23} are
open loop approaches, where interpolation and tensor completion are
sequentially performed for different objectives under different models;
but a non-ideal interpolation may destroy the desired \ac{btd} structure.
Our preliminary work \cite{SunCheLuo:C24} proposed a closed loop
approach to integrate interpolation and tensor completion, where a
tensor structure guided interpolation problem was formulated, but
by how much the tensor structure may help the interpolation was not
theoretically clear.

This paper aims to construct a 3D data array representing the spectrum
information over a 2D area, based on sparse, scattered, and incomplete
spectrum measurements. A tensor model is formed, where each slice
of the tensor represents a power spectrum map of one frequency band.
To integrate interpolation with tensor completion, we perform interpolation
framed by the \ac{btd} tensor structure, and in addition, the interpolation
is regularized by the low-rank property of the slices of the tensor.
The \ac{btd} model captures the spatial correlation of the power
spectrum maps across different frequency bands under the large-scale
fading, and the small-scale frequency-selective fading is treated
as noise that can be mitigated by the regression to the local polynomial
interpolation model. The low-rank regularization is to enforce the
interpolation to be aware of the global data structure. Analytical
results are established to understand the gain for the interpolation
under the \ac{btd} model. It is found that, at low \ac{snr}, the
power spectrum map excited by narrowband sources is easier to construct
than the one excited by wideband sources, whereas, at high \ac{snr},
the one excited by a wideband source with a uniform spectrum is easier
to construct. The reconstruction performance under overlapping spectrum
for a two-source case is also analyzed. Our numerical experiments
demonstrate that the proposed scheme works much better in separating
two sources when they have overlapping spectrums as compared to state-of-the-art
schemes, and the proposed algorithm is stable under off-grid measurements.
An improvement of over $20$\% in reconstruction accuracy is observed,
regardless of whether the sensors measure the full or sparse spectrum.

To summarize, the following contributions are made:
\begin{itemize}
\item We propose an integrated interpolation and \ac{btd} approach for
constructing power spectrum maps. The method incorporates an interpolation
model with the \ac{btd} structure to exploit the spatial correlation
of the power spectrum maps. A nuclear norm regularization is used
to exploit the low-rank property of the power spectrum maps.
\item We develop an alternating regression and singular value thresholding
algorithm for the integrated interpolation and \ac{btd} problem.
\item We establish analytical justification on why the \ac{btd} structure
may enhance the interpolation for power spectrum map construction,
with several important theoretical insights obtained from the analysis.
\item Extensive numerical studies are conducted and show that the proposed
method surpasses existing methods in terms of signal source separation
and power spectrum map construction, and it is stable under off-grid
measurements and inhomogeneity of the measurement topology.
\end{itemize}

The rest of the paper is organized as follows. Section \ref{sec:System-Model}
establishes the signal model and tensor model. Section \ref{sec:Tensor-guided-Interpolation}
develops an integrated interpolation and \ac{btd} approach with a
matrix formulation of the proposed method is provided. An alternating
regression and singular value thresholding algorithm is developed
to solve the proposed method. Section \ref{sec:Performance-of-tensor-guided}
theoretically analyzes the error of the proposed method and compares
it with conventional local polynomial interpolation. It then numerically
discusses low-rank regularization and the performance of signal source
separation. Numerical results are presented in Section \ref{sec:Numerical-Results}
and conclusion is given in Section \ref{sec:Conclusion}.

\emph{Notation:} Vectors are written as bold italic letters $\bm{x}$,
matrices as bold capital italic letters $\bm{X}$, and tensors as
bold calligraphic letters $\bm{\mathcal{\bm{X}}}$. For a matrix $\bm{X}$,
$[\bm{X}]_{(i,j)}$ denotes the entry in the $i$th row and $j$th
column of $\bm{X}$. For a tensor $\bm{\mathcal{\bm{X}}}$, $[\bm{\mathcal{\bm{X}}}]_{(i,j,k)}$
denotes the entry under the index $(i,j,k)$. The symbol `$\circ$'
represents outer product, `$\otimes$' represents Kronecker product,
and $\|\cdot\|_{F}$ represents Frobenius norm. The notation $o(x)$
means $\text{lim}_{x\to0}o(x)/x\rightarrow0$, $\text{diag}(\bm{x})$
represents a diagonal matrix whose diagonal elements are the entries
of vector $\bm{x}$, and $\text{vec}(\bm{\bm{X}})$ denotes the vectorization
of matrix $\bm{\bm{X}}$. The symbol $\mathbb{E}\left\{ \cdot\right\} $
and $\mathbb{V}\left\{ \cdot\right\} $ denote expectation and variance
separately.

\section{System Model\label{sec:System-Model}}

\subsection{Signal Model}

Consider that a bounded area $\mathcal{D}\subseteq\mathbb{R}^{2}$
contains $R$ signal sources. The signals occupying $K$ frequency
bands are detected by $M$ sensors at known locations $\bm{z}_{m}\in\mathcal{D}$,
$m=1,2,\dots,M$. Denote $\bm{s}_{r}\in\mathcal{D}$ as the location
of the $r$th source. Then, the signal power from the $r$th signal
source measured at the $k$th frequency band and location $\bm{z}$
is modeled as 
\begin{equation}
\rho_{k}^{(r)}(\bm{z})=\left(g_{r}(d(\bm{s}_{r},\bm{z}))+\zeta_{r}(\bm{z})+\eta_{r,k}(\bm{z})\right)\phi_{k}^{(r)}\label{eq:propagation model}
\end{equation}
where $g_{r}(d(\bm{s}_{r},\bm{z}))$ describes the path gain of the
$r$th source at distance $d(\bm{s}_{r},\bm{z})$, the function $d(\bm{s},\bm{z})=\|\bm{s}-\bm{z}\|_{2}$
describes the distance between a source at $\bm{s}$ and a sensor
at $\bm{z}$, $\zeta_{r}(\bm{z})$ captures the shadowing of the signal
from the $r$th source, $\eta_{r,k}(\bm{z})\sim\mathcal{N}(0,\sigma_{\eta}^{2})$
is a zero mean Gaussian random variable to model the fluctuation due
to the frequency-selective fading, and $\phi_{k}^{(r)}$ describes
the power allocation of the $r$th source at the $k$th frequency
band. Note that the values of all these components are not known to
the system.

The aggregated power at the $k$th frequency band from all the $R$
sources measured by a sensor located at $\bm{z}_{m}$ is denoted as
\begin{equation}
\gamma_{m}^{(k)}=\sum_{r=1}^{R}\rho_{k}^{(r)}(\bm{z}_{m})+\epsilon,\quad\forall k\in\Omega_{m}\label{eq:measurement model-non selective}
\end{equation}
where $\epsilon\sim\mathcal{N}(0,\sigma_{\epsilon}^{2})$ is to model
the measurement noise at each frequency band, and $\Omega_{m}\subseteq\{1,2,\dots,K\}$
contains the set of frequency bands that are measured by the $m$th
sensor. We assume that for each source $r$, the total power sums
to $\sum_{k=1}^{K}\phi_{k}^{(r)}=K$, and therefore, the \ac{snr}
is normalized to $K/(K\sigma_{\epsilon}^{2})=1/\sigma_{\epsilon}^{2}$,
where $K\sigma_{\epsilon}^{2}$ is the total noise for the entire
bandwidth.

Consider to discretize the target area $\mathcal{D}$ into $N_{1}$
rows and $N_{2}$ columns that results in $N_{1}\times N_{2}$ grid
cells. Let $\bm{c}_{ij}\in\mathcal{D}$ be the center location of
the $(i,j)$th grid cell. Our goal is to reconstruct the large-scale
propagation field, \emph{i.e.}, the first two terms in (\ref{eq:propagation model})
\begin{equation}
\rho^{(r)}(\bm{z})=g_{r}(d(\bm{s}_{r},\bm{z}))+\zeta_{r}(\bm{z})\label{eq:propagation-field-source-r}
\end{equation}
 at grid points $\bm{c}_{ij}$ and the power spectrum $\phi_{k}^{(r)}$,
$k=1,2,\dots,K$. The reconstructed propagation field model (\ref{eq:propagation-field-source-r})
does not capture the frequency-selective fading $\eta_{r,k}$.

As a result, from the propagation model (\ref{eq:propagation model}),
the measurement model in (\ref{eq:measurement model-non selective})
can be derived as
\begin{equation}
\gamma_{m}^{(k)}=\sum_{r=1}^{R}\rho^{(r)}(\bm{z}_{m})\phi_{k}^{(r)}+\tilde{\epsilon}_{k}\label{eq: measurement model}
\end{equation}
where $\tilde{\epsilon}_{k}=\text{\ensuremath{\sum_{r}}}\eta_{r,k}(\bm{z})\phi_{k}^{(r)}+\epsilon$
is a zero mean random variable that combines the randomness due to
the frequency-selective small-scale fading $\eta_{r,k}(\bm{z})\phi_{k}^{(r)}$
and the measurement noise $\epsilon$.

\subsection{Tensor Model\label{subsec:Tensor-Model}}

Let $\bm{S}_{r}\in\mathbb{R}^{N_{1}\times N_{2}}$ be a discretized
form of the propagation field for the $r$th source, where the $(i,j)$th
entry is given by $[\bm{S}_{r}]_{(i,j)}=\rho^{(r)}(\bm{c}_{ij})$.
It has been widely discussed in the literature that for many common
propagation scenarios, the matrix $\bm{S}_{r}$ tends to be low-rank
\cite{RomKim:J22,ZhaFuWan:J20}.

Let $\bm{\mathcal{H}}\in\mathbb{R}^{N_{1}\times N_{2}\times K}$ be
a tensor representation of the \emph{target} power spectrum maps to
be constructed. Based on (\ref{eq:propagation model}) and (\ref{eq:propagation-field-source-r}),
we have $[\bm{\mathcal{H}}]_{(i,j,k)}=\sum_{r=1}^{R}\rho^{(r)}(\bm{c}_{ij})\phi_{k}^{(r)}$
to represent the aggregated power of the $k$th frequency band from
the $R$ sources measured at location $\bm{c}_{ij}$ exempted from
the small-scale fading component $\eta_{r,k}(\bm{z})\phi_{k}^{(r)}$.
Denote $\bm{\phi}^{(r)}=[\phi_{1}^{(r)},...,\phi_{K}^{(r)}]^{\text{T}}$
as the power spectrum from the $r$th source. As a result, the tensor
$\bm{\mathcal{H}}$ has the following \ac{btd} structure
\begin{equation}
\bm{\mathcal{H}}=\text{\ensuremath{\sum_{r=1}^{R}}}\bm{S}_{r}\circ\bm{\bm{\phi}}^{(r)}\label{eq: BTD model}
\end{equation}
where `$\circ$' represents outer product.

Conventional tensor-based power spectrum map construction approaches
obtain the complete tensor $\bm{\mathcal{H}}$ from the measurement
$\gamma_{m}^{(k)}$ assuming that $\gamma_{m}^{(k)}$ are taken at
the center of the grid cell without measurement noise \cite{SchCavSta:C19,SheWanDin:J22}.
However, when the grid cells are too large, corresponding to small
$N_{1}$ and $N_{2}$, it is hard to guarantee that the sensor at
$\bm{z}_{m}$ is placed at the corresponding grid center $\bm{c}_{ij}$,
resulting in possibly large discretization error. When the grid cells
are small, corresponding to large $N_{1}$ and $N_{2}$, there might
be an identifiability issue as the dimension of the tensor is large.

Recent attempts \cite{SunChe:J22,CheWanZha:J23,SunChe:C22} consider
to first estimate $[\bm{\mathcal{H}}]_{(i,j,k)}$ using interpolation
methods based on the off-grid measurements, and then, employ matrix
completion or tensor completion based on the \ac{btd} model (\ref{eq: BTD model})
to improve the spectrum map construction. However, these methods are
open-loop methods where the property that $\bm{S}_{r}$ are low-rank
matrices is not exploited in the interpolation step; consequently,
a poor open-loop interpolation may affect the performance in the tensor
completion step.

\section{Integrated Interpolation and \ac{btd} Approach \label{sec:Tensor-guided-Interpolation}}

In this section, we propose an integrated interpolation and \ac{btd}
approach, where the \ac{btd} structure of the tensor model and the
low-rank property of the tensor components are both exploited for
interpolation.

\subsection{The Integrated Interpolation and \ac{btd} Problem}

Based on the \ac{btd} model in (\ref{eq: BTD model}), we consider
to fit a model $f^{(r)}(\bm{z})$ to approximate the large-scale propagation
field $\rho^{(r)}(\bm{z})$ of the $r$th source in (\ref{eq:propagation-field-source-r})
from the multi-band measurements $\gamma_{m}^{(k)}$ by exploiting
the structure of the tensor model $\text{\ensuremath{\bm{\mathcal{H}}}}$
and the low-rank property of the tensor component $\bm{S}_{r}$.

Here, we adopt a polynomial model for the propagation field $f^{(r)}(\bm{z})$
of the $r$th source, but note that, other conventional interpolation
approaches, including Kriging and kernel regression, also work with
the proposed framework. The global model $f^{(r)}(\bm{z})$ for source
$r$ can be constructed based on a number of \emph{local} models $f_{ij}^{(r)}(\bm{z})$
on selected cells $(i,j)\in\mathcal{I}$. \Ac{wlog}, a second order
polynomial model for the $(i,j)$th grid cell centered at $\bm{c}_{ij}$
is given as follows:
\begin{align}
f_{ij}^{(r)}(\bm{z})= & \alpha_{ij}^{(r)}+(\bm{\beta}_{ij}^{(r)})^{\text{T}}(\bm{z}-\bm{c}_{ij})\nonumber \\
 & \quad+(\bm{z}-\bm{c}_{ij})^{\text{T}}\text{\ensuremath{\text{\ensuremath{\text{\ensuremath{\bm{B}}}}}}}_{ij}^{(r)}(\bm{z}-\bm{c}_{ij}).\label{eq:2nd order LPR model}
\end{align}
We collect the model parameters of the local model at cell $(i,j)$
for the $r$th source into a vector $\bm{\theta}_{ij}^{(r)}=[\alpha_{ij}^{(r)},(\bm{\beta}_{ij}^{(r)})^{\text{T}},(\text{vec}(\bm{B}_{ij}^{(r)}))^{\text{T}}]^{\text{T}}\in\mathbb{R}^{7}$,
and $\bm{\Theta}_{ij}=[\bm{\theta}_{ij}^{(1)};\cdots;\bm{\theta}_{ij}^{(R)}]\in\mathbb{R}^{7R}$
is a collection of model parameters at cell $(i,j)$ for all sources.

It follows that, under perfect interpolation that yields $f_{ij}^{(r)}(\bm{c}_{ij})=\rho^{(r)}(\bm{c}_{ij})$,
we have $[\bm{\mathcal{H}}]_{(i,j,k)}=\sum_{r=1}^{R}f_{ij}^{(r)}(\bm{c}_{ij})\phi_{k}^{(r)}$,
which aligns with the \ac{btd} tensor structure in (\ref{eq: BTD model}).
Therefore, a \ac{btd} tensor structure guided least-squares local
polynomial regression based on the measurements $\gamma_{m}^{(k)}$
at the location $\bm{z}_{m}$ can be constructed through minimizing
the following cost
\begin{align}
 & l_{ij}(\bm{\Theta}_{ij},\{\phi_{k}^{(r)}\})\nonumber \\
 & \text{\ensuremath{\qquad}}=\sum_{m=1}^{M}\sum_{k\in\mathcal{\text{\ensuremath{\Omega}}}_{m}}\Big(\gamma_{m}^{(k)}-\sum_{r=1}^{R}f_{ij}^{(r)}(\bm{z}_{m})\phi_{k}^{(r)}\Big)^{2}\kappa_{ij}(\bm{z}_{m}).\label{eq:lij}
\end{align}
The term $\kappa_{ij}(\bm{z})\triangleq\kappa((\bm{z}-\bm{c}_{ij})/b)$
is a kernel function with a parameter $b$ to weight the importance
of the measurements. A possible choice of kernel function can be
the Epanechnikov kernel $\kappa(\bm{u})=\max\{0,\frac{3}{4}(1-||\bm{u}||^{2})\}$
\cite{Fan:b96}.

In an extreme case, local interpolation can be performed individually
for all cells where $\mathcal{I}$ includes all $N_{1}\times N_{2}$
entries, but the computation complexity can be very high when $N_{1}$
and $N_{2}$ are large. Therefore, we are interested in the case of
interpolating only a few (sparsely) selected grid cells due the spatial
correlation of the propagation field. As a result, the cost function
for the global model $f$ can be written as follows:
\begin{equation}
l(f)=\sum_{(i,j)\in\mathcal{I}}l_{ij}(\bm{\Theta}_{ij},\{\phi_{k}^{(r)}\}).\label{eq:l(f)}
\end{equation}

The regression model (\ref{eq:l(f)}) is not aware of the hidden low-rank
property of the propagation filed. We thus propose an integrated interpolation
and \ac{btd} formulation to impose the low-rank for the global model
$f$ as follows:
\begin{align}
\text{\ensuremath{\underset{\{\bm{\Theta}_{ij},\{\phi_{k}^{(r)}\}\},\{\bm{S}_{r}\}}{\text{minimize}}}} & l(f)+\frac{\nu}{2}\sum_{(i,j)\in\mathcal{I}}\sum_{r=1}^{R}(f_{ij}^{(r)}(\bm{c}_{ij})-[\bm{S}_{r}]_{(i,j)})^{2}\nonumber \\
 & \qquad\qquad\qquad\qquad\qquad+\text{\ensuremath{\mu\sum_{r=1}^{R}\|\bm{S}_{r}\|_{*}}}\label{eq: BTD interp regularized}
\end{align}
where $\|\cdot\|_{*}$ represents the nuclear norm. As a result, the
regression model $f$ not only needs to fit the measurement data $\gamma_{m}^{(k)}$
via minimizing the cost $l_{f}(\cdot)$ in (\ref{eq:l(f)}) but is
also penalized by the rank of $\bm{S}_{r}$ via the second and the
third terms in (\ref{eq: BTD interp regularized}).

\subsection{Matrix Formulation of Integrated Interpolation and \ac{btd}}

It turns out that the minimization problem (\ref{eq: BTD interp regularized})
has a nice structure after reformulation with matrix representations.

Denote $\bm{x}_{m}(\bm{c}_{ij})=[1,(\bm{z}_{m}-\bm{c}_{ij})^{\text{T}},(\text{vec}((\bm{z}_{m}-\bm{c}_{ij})(\bm{z}_{m}-\bm{c}_{ij})^{\text{T}}))^{\text{T}}]^{\text{T}}\in\mathbb{R}^{7}$.
Then, the polynomial model $f_{ij}^{(r)}(\bm{z}_{m})$ in (\ref{eq:2nd order LPR model})
is rewritten as $f_{ij}^{(r)}(\bm{z}_{m})=\bm{x}_{m}^{\text{T}}(\bm{c}_{ij})\bm{\theta}_{ij}^{(r)}.$

Denote $\bm{\phi}_{k}=[\phi_{k}^{(1)},\cdots,\phi_{k}^{(R)}]^{\text{T}}\in\mathbb{R}^{R}$,
and recall the definition of $\bm{\Theta}_{ij}$ in (\ref{eq:lij}).
From (\ref{eq:lij}), we have $\sum_{r=1}^{R}f_{ij}^{(r)}(\bm{z}_{m})\phi_{k}^{(r)}=\bm{\phi}_{k}^{\text{T}}\otimes\bm{x}_{m}^{\text{T}}(\bm{c}_{ij})\bm{\Theta}_{ij}$
where `$\otimes$' is the Kronecker product. The least-squares local
polynomial regression model (\ref{eq:lij}) is rewritten as follows:
\begin{align}
 & l_{ij}(\bm{\Theta}_{ij},\{\phi_{k}^{(r)}\})\nonumber \\
 & =\sum_{m=1}^{M}\sum_{k\in\Omega_{m}}\Big(\gamma_{m}^{(k)}-\bm{\phi}_{k}^{\text{T}}\otimes\bm{x}_{m}^{\text{T}}(\bm{c}_{ij})\bm{\Theta}_{ij}\Big)^{2}\kappa_{ij}(\bm{z}_{m}).\nonumber \\
 & =\sum_{m=1}^{M}\sum_{k=1}^{K}\psi_{mk}\Big(\gamma_{m}^{(k)}-\bm{\phi}_{k}^{\text{T}}\otimes\bm{x}_{m}^{\text{T}}(\bm{c}_{ij})\bm{\Theta}_{ij}\Big)^{2}\kappa_{ij}(\bm{z}_{m})\label{eq:lij kronecker}
\end{align}
where $\psi_{mk}$ indicates the observation status of $k$th frequency
band at location $m$, if $k\in\Omega_{m}$, we have $\psi_{mk}=1$,
elsewise $\psi_{mk}=0$. For the convenience of matrix form formulation,
we denote $[\bm{\psi}]_{mk}=\psi_{mk}$.

Next, we further rearrange (\ref{eq:lij kronecker}) into a matrix
form. Denote $\bm{X}_{ij}=[\bm{x}_{1}(\bm{c}_{ij}),\bm{x}_{2}(\bm{c}_{ij}),...,\bm{x}_{M}(\bm{c}_{ij})]\in\mathbb{R}^{7\times M}$.
To concisely express the mathematical relationships, we first replace
the summation term over $m$ in (\ref{eq:lij kronecker}) with its
matrix form as follows:
\begin{align}
 & \sum_{m=1}^{M}\sum_{k=1}^{K}\psi_{mk}\Big(\gamma_{m}^{(k)}-\bm{\phi}_{k}^{\text{T}}\otimes\bm{x}_{m}^{\text{T}}(\bm{c}_{ij})\bm{\Theta}_{ij}\Big)^{2}\kappa_{ij}(\bm{z}_{m})\nonumber \\
= & \sum_{m=1}^{M}\sum_{k=1}^{K}\Big(\psi_{mk}\sqrt{\kappa_{ij}(\bm{z}_{m})}\left(\gamma_{m}^{(k)}-\bm{\phi}_{k}^{\text{T}}\otimes\bm{x}_{m}^{\text{T}}(\bm{c}_{ij})\bm{\Theta}_{ij}\right)\Big)^{2}\nonumber \\
= & \sum_{k=1}^{K}\left\Vert \mathbf{Q}_{ij}\text{diag}(\bm{\psi}(:,k))\left(\bm{\Gamma}(:,k)-\bm{\phi}_{k}^{\text{T}}\otimes\bm{X}_{ij}^{\text{T}}\bm{\Theta}_{ij}\right)\right\Vert _{2}^{2}\label{eq:lij matrix form}
\end{align}
where $\bm{\Gamma}(m,k)=\gamma_{m}^{(k)}$, $\bm{\Gamma}(:,k)$ and
$\bm{\psi}(:,k)$ respectively represent the $k$th column of $\bm{\Gamma}$
and $\bm{\psi}$. The matrix $\mathbf{Q}_{ij}$ is an $M\times M$
diagonal matrix with the $m$th diagonal element equals to $\ensuremath{\sqrt{\kappa_{ij}(\bm{z}_{m})}}$.

Further, based on (\ref{eq:lij matrix form}), we replace the summation
term over $k$ with its matrix form. Denote $\bm{\Phi}=\text{[\ensuremath{\bm{\phi}_{1}},\ensuremath{\bm{\phi}_{2}},\ensuremath{\cdots,\bm{\phi}_{K}}]\ensuremath{\in\mathbb{R}}}^{R\times K}$.
Then, model (\ref{eq:lij matrix form}) can be written as follows:
\begin{align*}
 & \sum_{k\in\Omega_{m}}\left\Vert \mathbf{Q}_{ij}\text{diag}(\bm{\psi}(:,k))\left(\bm{\Gamma}(:,k)-\bm{\phi}_{k}^{\text{T}}\otimes\bm{X}_{ij}^{\text{T}}\bm{\Theta}_{ij}\right)\right\Vert _{2}^{2}\\
 & \qquad\qquad=\left\Vert \bm{W}_{ij}\left(\text{vec}(\bm{\bm{\Gamma}})-\bm{\bm{\Phi}}^{\text{T}}\otimes\bm{X}_{ij}^{\text{T}}\bm{\Theta}_{ij}\right)\right\Vert _{2}^{2}
\end{align*}
where $\bm{W}_{ij}$ is an $MK\times MK$ diagonal matrix, 
\begin{equation}
\bm{W}_{ij}=\bm{\mathbf{I}}_{K}\otimes\mathbf{Q}_{ij}\text{diag}(\text{vec}(\bm{\psi}))\in\mathbb{R}^{MK\times MK}\label{eq:weights}
\end{equation}
and $\bm{\mathbf{I}}_{K}\in\mathbb{R}^{K\times K}$ is an identity
matrix.

From (\ref{eq:2nd order LPR model}), if $\bm{z}=\bm{c}_{ij}$, then
we have $f_{ij}^{(r)}(\bm{c}_{ij})=\alpha_{ij}^{(r)}=\bm{e}_{r}^{\text{T}}\bm{\Theta}_{ij}$,
where $\bm{e}_{r}$ is a unit vector with the $(7\times(r-1)+1)$th
entry equals $1$, and all the other elements equal $0$. Thus, the
integrated interpolation and \ac{btd} problem (\ref{eq: BTD interp regularized})
can be rewritten as follows:
\begin{align}
\text{\ensuremath{\underset{\bm{\Theta}_{ij},\bm{\bm{\Phi}},\{\bm{S}_{r}\}}{\text{minimize}}}} & \sum_{(i,j)\in\mathcal{I}}\left\Vert \bm{W}_{ij}\left(\text{vec}(\bm{\bm{\Gamma}})-\bm{\bm{\Phi}}^{\text{T}}\otimes\bm{X}_{ij}^{\text{T}}\bm{\Theta}_{ij}\right)\right\Vert _{2}^{2}\nonumber \\
 & +\nu\sum_{(i,j)\in\mathcal{I}}\sum_{r=1}^{R}(\bm{e}_{r}^{\text{T}}\bm{\Theta}_{ij}-[\bm{S}_{r}]_{(i,j)})^{2}+\text{\ensuremath{\mu\sum_{r=1}^{R}\|\bm{S}_{r}\|_{*}}}\nonumber \\
\text{\text{subject to}} & \ [\bm{\bm{\Phi}}]_{ij}\geq0\label{eq:matrix form BTD interpo regularized (full)}
\end{align}
where $\bm{\Gamma}$ is a collection of \ac{rss} $\gamma_{m}^{(k)}$,
$\bm{\bm{\Phi}}$ is the power spectrum which is non-negative, $\bm{W}_{ij}$
is the diagonal matrix of weights defined in (\ref{eq:weights}),
$\bm{X}_{ij}$ is a collection of measurement locations $\bm{z}_{m}$,
and $\bm{S}_{r}$ is a discretized form of propagation field in (\ref{eq: BTD model}).

\subsection{Alternating Regression and Singular Value Thresholding}

It is observed in (\ref{eq:matrix form BTD interpo regularized (full)})
that given the matrix components $\bm{S}_{r}$ and the spectrum variable
$\bm{\bm{\Phi}}$, the objective is a convex quadratic function in
the regression parameters $\bm{\Theta}_{ij}$. Likewise, given $\bm{S}_{r}$
and $\bm{\Theta}_{ij}$, (\ref{eq:matrix form BTD interpo regularized (full)})
is also a convex quadratic function in $\bm{\bm{\Phi}}$; and moreover,
(\ref{eq:matrix form BTD interpo regularized (full)}) is a convex
function in $\bm{S}_{r}$. Therefore, it is natural to adopt an alternating
optimization approach to solve for the integrated interpolation and
\ac{btd} problem (\ref{eq:matrix form BTD interpo regularized (full)}).

\textbf{Update of $\bm{\Theta}_{ij}$}: Given the values of $\bm{S}_{r}$
and $\bm{\bm{\Phi}}$, the optimization problem (\ref{eq:matrix form BTD interpo regularized (full)})
is equivalent to a weighted least-squares problem as follows:
\begin{align}
\text{\ensuremath{\underset{\bm{\Theta}_{ij}}{\text{minimize}}}} & \sum_{(i,j)\in\mathcal{I}}\left\Vert \bm{W}_{ij}\left(\text{vec}(\bm{\bm{\Gamma}})-\bm{\bm{\Phi}}^{\text{T}}\otimes\bm{X}_{ij}^{\text{T}}\bm{\Theta}_{ij}\right)\right\Vert _{2}^{2}\nonumber \\
 & \qquad+\nu\sum_{(i,j)\in\mathcal{I}}\sum_{r=1}^{R}(\bm{e}_{r}^{\text{T}}\bm{\Theta}_{ij}-[\bm{S}_{r}]_{(i,j)})^{2}.\label{eq:theta_ij least square}
\end{align}
Note that the problem (\ref{eq:theta_ij least square}) is an unconstrained
strictly convex problem. Hence, the solution can be obtained through
setting the first order derivative of (\ref{eq:theta_ij least square})
to zero and we get:
\begin{align}
\hat{\bm{\Theta}}_{ij} & =(\bm{\bm{\Phi}}\otimes\bm{X}_{ij}\bm{W}_{ij}^{2}\bm{\bm{\Phi}}^{\text{T}}\otimes\bm{X}_{ij}^{\text{T}}+\nu\sum_{r=1}^{R}\bm{e}_{r}\bm{e}_{r}^{\text{T}})^{-1}\nonumber \\
 & \ \ \ \times\big(\bm{\bm{\Phi}}\otimes\bm{X}_{ij}\bm{W}_{ij}^{2}\text{vec}(\bm{\bm{\Gamma}})+\nu\sum_{r=1}^{R}\bm{e}_{r}[\bm{S}_{r}]_{(i,j)}\big).\label{eq:sol theta}
\end{align}

\textbf{Update of $\bm{\Phi}$}: Similarly, given the values of $\bm{S}_{r}$
and $\bm{\Theta}_{ij}$, the optimization problem (\ref{eq:matrix form BTD interpo regularized (full)})
is equivalent to a constrained weighted least-squares as follows:
\begin{align}
\text{\ensuremath{\underset{\bm{\bm{\Phi}}}{\text{minimize}}}} & \sum_{(i,j)\in\mathcal{I}}\left\Vert \bm{W}_{ij}\left(\text{vec}(\bm{\bm{\Gamma}})-\bm{\bm{\Phi}}^{\text{T}}\otimes\bm{X}_{ij}^{\text{T}}\bm{\Theta}_{ij}\right)\right\Vert _{2}^{2}\nonumber \\
\text{subject to} & \ [\bm{\bm{\Phi}}]_{ij}\geq0.\label{eq:update phi}
\end{align}

Using the property of Kronecker product $\text{vec}(\bm{AXB})=(\bm{B}^{\text{T}}\otimes\bm{A})\text{vec}(\bm{X})$
and $\text{vec}(\bm{AB})=(\text{\ensuremath{\bm{I}}}\otimes\bm{A})\text{vec}(\bm{B})=(B^{\text{T}}\otimes\bm{I})\text{vec}(\bm{A})$
\cite{Gra:B18}, problem (\ref{eq:update phi}) can be rewritten as:
\begin{align}
\text{\ensuremath{\underset{\bm{\bm{\Phi}}}{\text{minimize}}}} & \sum_{(i,j)\in\mathcal{I}}\biggl|\biggl|\bm{W}_{ij}\big(\text{vec}(\bm{\bm{\Gamma}})\nonumber \\
 & \qquad\qquad-\left(\bm{I}_{K}\otimes(\bm{X}_{ij}^{\text{T}}[\bm{\Theta}_{ij}]_{7\times R})\right)\text{vec}(\bm{\bm{\Phi}})\big)\biggl|\biggl|_{2}^{2}\nonumber \\
\text{subject to} & \ [\bm{\bm{\Phi}}]_{ij}\geq0\label{eq:update phi-1-1}
\end{align}
where $[\bm{\Theta}_{ij}]_{7\times R}$ is the reformed matrix form
of vector $\bm{\Theta}_{ij}$ with dimension $7\times R$.

Note that the $m$th grid in set $\mathcal{I}$ is a one-to-one mapping
to the $(i,j)$th grid. Denote $\tilde{\bm{W}}_{ij}=[\bm{W}_{1};\cdots;\bm{W}_{m};\cdots;\bm{W}_{|\mathcal{I}|}]$,
$\bm{Z}_{ij}=\bm{W}_{ij}(\bm{I}_{K}\otimes(\bm{X}_{ij}^{\text{T}}[\bm{\Theta}_{ij}]_{7\times R}))$
and $\tilde{\bm{Z}}_{ij}=[\bm{Z}_{1};\cdots;\bm{Z}_{m};\cdots;\bm{Z}_{|\mathcal{I}|}]$.
We reformulate (\ref{eq:update phi-1-1}) to make it a general non-negative
constrained least-squares problem as follows:
\begin{align}
\text{\ensuremath{\underset{\bm{\bm{\Phi}}}{\text{minimize}}}} & \biggl|\biggl|\tilde{\bm{W}}_{ij}\left(\text{vec}(\bm{\bm{\Gamma}})-\tilde{\bm{Z}}_{ij}\text{vec}(\bm{\bm{\Phi}})\right)\biggl|\biggl|_{2}^{2}\nonumber \\
\text{subject to} & \ [\bm{\bm{\Phi}}]_{ij}\geq0.\label{eq:update phi nnls}
\end{align}

Note that the problem (\ref{eq:update phi nnls}) is a strictly convex
problem. Hence, it has a unique optimal solution and a single principal
pivoting algorithm \cite{PortJud:J94} can be applied to solve (\ref{eq:update phi nnls}).

\textbf{Update of $\bm{S}_{r}$}: Finally, with the $\bm{\Theta}_{ij}$
and $\bm{\bm{\Phi}}$ updated through (\ref{eq:sol theta}) and (\ref{eq:update phi nnls}),
the optimization problem (\ref{eq:matrix form BTD interpo regularized (full)})
is equivalent to a nuclear norm regularized low-rank matrix completion
problem as follows:
\begin{align}
\text{\ensuremath{\underset{\{\bm{S}_{r}\}}{\text{minimize}}}} & \sum_{(i,j)\in\mathcal{I}}\sum_{r=1}^{R}(\bm{e}_{r}^{\text{T}}\bm{\Theta}_{ij}-[\bm{S}_{r}]_{(i,j)})^{2}+\text{\ensuremath{\mu\sum_{r=1}^{R}\|\bm{S}_{r}\|_{*}.}}\label{eq:svt}
\end{align}
 It is observed that (\ref{eq:svt}) can be equivalently decomposed
into $R$ parallel sub-problems each focusing on an $\bm{S}_{r}$
as follows:
\begin{align}
\text{\ensuremath{\underset{\{\bm{S}_{r}\}}{\text{minimize}}}} & \sum_{(i,j)\in\mathcal{I}}(\bm{e}_{r}^{\text{T}}\bm{\Theta}_{ij}-[\bm{S}_{r}]_{(i,j)})^{2}+\text{\ensuremath{\mu\|\bm{S}_{r}\|_{*}.}}\label{eq:svt-1}
\end{align}
Denote the observation matrix $\bm{\Psi}$ as $[\bm{\Psi}]_{ij}=\bm{e}_{r}^{\text{T}}\bm{\Theta}_{ij}$.
Then, the singular value thresholding algorithm can be applied to
solve (\ref{eq:svt-1}) through the iteration as follows \cite{CaiCan:J10}:
\begin{equation}
\begin{cases}
\bm{S}_{r}^{(k)}=\mathcal{S}_{\mu}(\bm{Y}^{(k-1)})\\
\bm{Y}^{(k)}=\bm{Y}^{(k-1)}+\text{\ensuremath{\delta\mathcal{P}}}_{\Omega}(\bm{\Psi}-\bm{S}_{r}^{(k)})
\end{cases}\label{eq:svt sol}
\end{equation}
where $\bm{Y}$ is an intermediate matrix, $(k)$ represents the index
of each iteration, $[\ensuremath{\mathcal{P}}_{\Omega}(\bm{X})]_{ij}=[\bm{X}]_{ij}$
if $(i,j)\in\Omega$ and zero otherwise, and $\mathcal{S}_{\mu}$
is the soft-thresholding operator which is defined as follows:
\[
\mathcal{S}_{\mu}(\bm{Y}^{(k-1)})=\bm{U}_{\iota}\mathcal{D}_{\mu}(\bm{\Sigma}_{\iota})\bm{V}_{\iota}^{\text{T}}
\]
with $\mathcal{D}_{\mu}(\bm{\Sigma}_{\iota})=\text{diag}[(\sigma_{1}-\mu)_{+},\cdots(\sigma_{\iota}-\mu)_{+}]$,
$\bm{\Sigma}_{\iota}=\text{diag}[\sigma_{1},\cdots,\sigma_{\iota}]$,
$(x)_{+}=\text{max}(0,x)$, and $\bm{U}_{\iota}\bm{\Sigma}_{\iota}\bm{V}_{\iota}^{\text{T}}$
is the \ac{svd} of $\bm{Y}^{(k-1)}$ where $\iota$ is the rank of
$\bm{Y}^{(k-1)}$.

Each sub-problem of (\ref{eq:matrix form BTD interpo regularized (full)})
is strictly convex and has a unique solution. Since the overall objective
value decreases (or at least does not increase) with each iteration,
and the objective (\ref{eq: BTD interp regularized}) is lower bounded
by $0$, the alternating regression and singular value thresholding
algorithm must converge.

Finally, the power spectrum map is constructed as $\hat{\bm{\mathcal{H}}}=\sum_{r=1}^{R}\hat{\bm{S}_{r}}\circ\hat{\bm{\phi}}^{(r)}$.

\section{Performance of the Integrated Interpolation and \ac{btd} Approach\label{sec:Performance-of-tensor-guided}}

This section investigates the performance and the potential advantage
of the proposed integrated interpolation and \ac{btd} approach (\ref{eq: BTD interp regularized})
from three aspects. First, we focus on the least-squares cost term
$l(f)$ by dropping the low-rank regularization in (\ref{eq: BTD interp regularized}).
By analyzing the interpolation error, we show that the proposed method
indeed yields a better accuracy. Then, we numerically verify that
by imposing the low-rank regularization, the reconstruction accuracy
can be further improved. Finally, we discuss how the proposed structure
may improve the identifiability of separating multiple sources compared
to the classical tensor-based approaches.

\subsection{Improvement from the \ac{btd} Model}

We first show that exploiting the \ac{btd} model in (\ref{eq: BTD model})
can increase the accuracy of constructing $\bm{S}_{r}$.

\subsubsection{Single Source}

Let us start from the case of full spectrum observation for a single
source, where $R=1$ and each sensor measures all the frequency bands,
i.e., $\bm{\psi}$ in (\ref{eq:lij matrix form}) is an all $1$'s
matrix. The result can be easily extended to the case where each sensor
only observes a subset of frequency bands.

For the sake of notation simplicity, we omit the superscript \textquotedbl (r)\textquotedbl{}
and adopt symbols $\rho$, $f$, and so on, to represent $\rho^{(r)}$,
$f^{(r)}$, and similar variables. Recall that $f_{ij}(\bm{c}_{ij})$
is a model to approximate $\rho(\bm{c}_{ij})$, and $\alpha_{ij}=f_{ij}(\bm{c}_{ij})$
from (\ref{eq:2nd order LPR model}). Define the interpolation error
as $\xi_{ij}=\hat{\alpha}_{ij}-\rho(\bm{c}_{ij})$ where $\hat{\alpha}_{ij}$
is an estimate of $\alpha_{ij}$.

Assume that the propagation field $\rho_{k}(\bm{z})$ in (\ref{eq:propagation model})
is third-order differentiable. Then, if $\phi_{k}$ is available,
we have the following results to characterize the construction error
$\xi_{ij}$ at a point $\bm{c}_{ij}$ of the integrated interpolation
and \ac{btd} approach.
\begin{thm}[Full spectrum error analysis]
\label{thm:tensor-guided interpolation (full)} Let $\hat{\alpha}_{ij}$
be the estimate obtained from the solution to the integrated interpolation
and \ac{btd} problem (\ref{eq:theta_ij least square}) under $\nu=0$
based on the regression model (\ref{eq:2nd order LPR model}). The
variance $\mathcal{E}_{\text{t}}=\mathbb{V}\{\xi_{ij}\}$ of the interpolation
error $\xi_{ij}=\hat{\alpha}_{ij}-\rho(\bm{c}_{ij})$ is given by
\begin{align}
\mathcal{E}_{\text{t}} & =\left(\frac{\sum_{k=1}^{K}\text{\ensuremath{\phi_{k}^{4}}}}{\left(\sum_{k=1}^{K}\text{\ensuremath{\phi_{k}^{2}}}\right)^{2}}\sigma_{\eta}^{2}+\frac{1}{\sum_{k=1}^{K}\text{\ensuremath{\phi_{k}^{2}}}}\sigma_{\epsilon}^{2}\right)C\left(\{\bm{z}_{m}\},b\right)\label{eq:thm-variance-tensor}
\end{align}
where $C(\{\bm{z}_{m}\},b)$ is a constant that depends on the sensor
locations $\bm{z}_{m}$ and the window size $b$.
\end{thm}
\begin{proof}
See Appendix \ref{sec:Proof-of-Theorem k bands}.
\end{proof}
It is observed that the variance of the interpolation error $\xi_{ij}$
depends on the spectrum $\bm{\bm{\Phi}}$ of the source. An analysis
of the coefficients is given in the following proposition.
\begin{prop}[Impact from the spectrum]
\label{prop:Impact-from-the-spectrum}The coefficients in (\ref{eq:thm-variance-tensor})
satisfy
\[
\frac{1}{K}\leq\varpi_{\eta}(\bm{\Phi})\triangleq\frac{\sum_{k=1}^{K}\text{\ensuremath{\phi_{k}^{4}}}}{\left(\sum_{k=1}^{K}\text{\ensuremath{\phi_{k}^{2}}}\right)^{2}}\leq1
\]
with lower bound achieved when $\bm{\Phi}=(1,1,\dots,1)$ and upper
bound achieved when $\bm{\bm{\Phi}}=(K,0,0,\dots,0)$. In addition,
\[
\frac{1}{K^{2}}\leq\varpi_{\epsilon}(\bm{\bm{\Phi}})\triangleq\frac{1}{\sum_{k}\phi_{k}^{2}}\leq\frac{1}{K}
\]
with lower bound achieved when $\bm{\Phi}=(K,0,0,\dots,0)$ and upper
bound achieved when $\bm{\bm{\Phi}}=(1,1,\dots,1)$.
\end{prop}
\begin{proof}
See Appendix \ref{sec:Proof-of-Proposition impact from the spectrum for tensor}.
\end{proof}
One can make the following observations. First, the impact from the
frequency-selective fading due to the power allocation of the source
is more significant than that from the measurement noise, in the sense
that for $\sigma_{\eta}=\sigma_{\epsilon}$, we have $\varpi_{\epsilon}(\bm{\Phi})\leq\varpi_{\eta}(\bm{\bm{\Phi}})$
for all power spectrum $\bm{\Phi}$. The intuition is that the frequency-selective
fading component $\eta_{r,k}$ is multiplied with the power allocation
$\phi_{k}$ in (\ref{eq:propagation model}) for the contribution
to the measurement in (\ref{eq: measurement model}), and hence, a
power boost in the $k$th frequency band also enhances the frequency-selective
fading component.

Second, it follows that equal-power allocation $\bm{\Phi}=(1,1,\dots,1)$
minimizes the impact of the frequency-selective fading under the integrated
interpolation and \ac{btd} approach. In this case, the fading component
can be treated as additional measurement noise, resulting in a combined
noise with variance $\sigma_{\eta}^{2}+\sigma_{\epsilon}^{2}$. The
more frequency bands that can be measured, the less impact from the
frequency-selective fading.

Third, on the contrary to the impact of the frequency-selective fading,
the noise term prefers a power boost in just one frequency band. This
is because reducing the bandwidth also reduces the noise power, resulting
in an increase in the \ac{snr}.

To summarize, in the high \ac{snr} scenario, where the measurement
noise power $\sigma_{\epsilon}^{2}$ is much smaller than the amount
of the frequency-selective fading $\sigma_{\eta}^{2}$, a wide-band
measurement with equal power allocation is preferred for constructing
the propagation map $\bm{S}_{r}$. On the contrary, in the low \ac{snr}
scenario, a narrow-band measurement is preferred and all the power
should be allocated to a single frequency band. This is mathematically
summarized in the following corollary.
\begin{cor}
[Asymptotic performance]At high \ac{snr}, the uniform spectrum
$\bm{\bm{\Phi}}=(1,1,\dots,1)$ asymptotically minimizes $\mathcal{E}_{\text{t}}$
as $1/\sigma_{\epsilon}^{2}\to\infty$, and consequently, $\mathcal{E}_{\text{t}}\to C\sigma_{\eta}^{2}/K$.
At low \ac{snr}, the single band measurement under $\bm{\bm{\Phi}}=(K,0,0,\dots,0)$
asymptotically minimizes $\mathcal{E}_{\text{t}}$ as $1/\sigma_{\epsilon}^{2}\to0$,
and consequently, $\mathcal{E}_{\text{t}}\to C\sigma_{\epsilon}^{2}/K^{2}$.
\end{cor}
For performance bench-marking, we consider a conventional frequency-by-frequency
construction for each $\rho_{k}(\bm{z})$, using a local polynomial
interpolation technique under the same regression model (\ref{eq:2nd order LPR model}).
Specifically, for each $k$, we construct $\rho_{k}(\bm{z})$ only
based on $\bm{z}_{m}$ and the measurement of the $k$th frequency
band $\gamma_{m}^{(k)}$ using local polynomial regression under the
same parameter $b$ and the kernel function. Likewise, the spectrum
$\bm{\bm{\Phi}}$ is available. Denote $\hat{\alpha}_{ij}^{(k)}$
as the estimated power spectrum $\rho_{k}(\bm{c}_{ij})$ at location
$\bm{c}_{ij}$.
\begin{prop}[Interpolation without a BTD model]
\label{pro:traditional LP}Let $\hat{\alpha}_{ij}=\sum_{k=1}^{K}\hat{\alpha}_{ij}^{(k)}/K$,
and $\xi_{ij}=\sum_{k=1}^{K}\xi_{ij}^{(k)}/K$ where $\xi_{ij}^{(k)}=\hat{\alpha}_{ij}^{(k)}-\rho_{k}(\bm{c}_{ij})$.
Under the same assumption as in Theorem~\ref{thm:tensor-guided interpolation (full)},
the variance $\mathcal{E}_{\text{p}}=\mathbb{V}\{\xi_{ij}\}$ of the
interpolation error $\xi_{ij}$ is given by
\begin{align}
\mathcal{E}_{\text{p}} & =\left(\sigma_{\text{\ensuremath{\eta}}}^{2}+\sum_{k=1}^{K}\frac{1}{\text{\ensuremath{\phi_{k}^{2}K}}}\sigma_{\epsilon}^{2}\right)C\left(\{\bm{z}_{m}\},b\right).\label{eq:thm-variance-lpr}
\end{align}
\end{prop}
\begin{proof}
See Appendix \ref{sec:Proof-of-Proposition LP}.
\end{proof}
It is observed that the impact from the frequency-selective fading
$\sigma_{\text{\ensuremath{\eta}}}^{2}$ is independent of the spectrum
$\bm{\Phi}$, which is in contrast to (\ref{eq:thm-variance-tensor})
for the case with the tensor guidance. This is due to the fact that
only the data $\gamma_{m}^{(k)}$ for the same frequency band is used.
The coefficient $\sum_{k=1}^{K}1/(\phi_{k}^{2}K)$ of measurement
noise $\sigma_{\epsilon}^{2}$ in (\ref{eq:thm-variance-lpr}) is
lower bounded by $1$ with lower bound achieved when $\bm{\Phi}=(1,1,\dots,1)$.
Note that this error coefficient $\sum_{k=1}^{K}1/(\phi_{k}^{2}K)$
can be arbitrarily large for an arbitrarily small $\phi_{k}$ in a
particular frequency band $k$. On the contrary, the construction
performance for the integrated interpolation and \ac{btd} approach
in (\ref{eq:thm-variance-tensor}) is less affected even $\phi_{k}=0$
for some $k$. 

To make a more specific comparison, we derive the difference $\mathcal{E}_{\text{p}}-\mathcal{E}_{\text{t}}$.
\begin{prop}[Interpolation error reduction]
\label{pro:difference of variance}Under the same assumption as in
Theorem 1, the proposed integrated interpolation and \ac{btd} approach
reduces the variance of the interpolation error by
\begin{align}
 & \mathcal{E}_{\text{p}}-\mathcal{E}_{\text{t}}=\left(\varsigma_{\eta}(\bm{\Phi})\sigma_{\eta}^{2}+\varsigma_{\epsilon}(\bm{\Phi})\sigma_{\epsilon}^{2}\right)C\left(\{\bm{z}_{m}\},b\right)\label{eq:difference of variance}
\end{align}
where $\varsigma_{\eta}(\bm{\Phi})=2(\sum_{i\neq j}\text{\ensuremath{\phi_{i}^{2}}}\text{\ensuremath{\phi_{j}^{2}}})/(\sum_{k=1}^{K}\text{\ensuremath{\phi_{k}^{2}}})^{2}$\textup{
and} $\varsigma_{\epsilon}(\bm{\Phi})=\sum_{k=1}^{K}((\sum_{l\neq k}^{K}\text{\ensuremath{\phi_{l}^{2}}})/(\phi_{k}^{2}K))/\sum_{k=1}^{K}\text{\ensuremath{\phi_{k}^{2}}}.$
\end{prop}
In addition, $0\leq\varsigma_{\eta}(\bm{\Phi})\leq\frac{K-1}{K}\leq\varsigma_{\epsilon}(\bm{\Phi})$
where the second and third inequalities are achieved when $\bm{\Phi}=(1,1,\dots,1)$
and the inequality in first inequality is achieved when $\bm{\Phi}=(K,0,0,\dots,0)$.
\begin{proof}
See Appendix \ref{sec:Proof-of-Corollary sparse tensor guided}.
\end{proof}
It is observed that the integrated interpolation and \ac{btd} approach
is always of smaller variance than the conventional frequency-by-frequency
local polynomial interpolation. In addition, the error from the measurement
noise is more significant than that from the frequency-selective fading,
in the sense that for $\sigma_{\eta}=\sigma_{\epsilon}$, we have
$\varsigma_{\eta}(\bm{\Phi})\sigma_{\eta}\leq\varsigma_{\epsilon}(\bm{\Phi})\sigma_{\epsilon}$,
because $\varsigma_{\eta}(\bm{\Phi})\leq\frac{K-1}{K}\leq\varsigma_{\epsilon}(\bm{\Phi})$.

Finally, we extend our discussion to the case of sparse spectrum observation.
We assume each sensor $m$ randomly measures $|\mathcal{\text{\ensuremath{\Omega}}}_{m}|=K'$
frequency bands. To simplify the discussion, we assume the source
has equal power allocation $\bm{\Phi}=(1,1,\dots,1)$ over all frequency
bands.
\begin{prop}[Sparse spectrum error analysis]
\label{prop:sparse spectrum interp}Under the same condition of Theorem
\ref{thm:tensor-guided interpolation (full)}, the variance $\mathcal{E}_{\text{t}}$
of the interpolation error $\xi_{ij}$ for the integrated interpolation
and \ac{btd} approach is given by
\[
\mathcal{E}_{\text{t}}=\frac{\sigma_{\eta}^{2}+\sigma_{\epsilon}^{2}}{K'}C(\{\bm{z}_{m}\},b).
\]
\end{prop}
\begin{proof}
See Appendix \ref{sec:Proof-of-Proposition sparse spectrum}.
\end{proof}
Proposition \ref{prop:sparse spectrum interp} verifies that the interpolation
performance for the integrated interpolation and \ac{btd} approach
depends on the number of observed frequency bands $K'$ in each sensor,
rather than the measurement pattern in the frequency domain, i.e.,
different sensors can observe different subsets of frequency bands.
For instance, we may have more measurements for some frequency bands,
while other frequency bands have rare measurements, and such a heterogeneous
situation does not affect the performance for the integrated interpolation
and \ac{btd} approach.

\subsubsection{Multiple Sources}

We demonstrate the result for a two-source case, where the derivation
for $R>2$ is straight-forward, and similar insights can be obtained.

Consider that the two sources occupy $K$ frequency bands. There are
$\eta K$ frequency bands that are occupied by both sources where
$0\leq\eta\leq1$ is the overlapping ratio, and $\frac{1-\eta}{2}K$
distinct frequency bands are occupied by each source separately. Assume
equal power allocation across frequency bands for each source.

Since the configuration is symmetric, it suffices to analyze the reconstruction
for the $r$th source, $r=1,2$, and we have $\xi_{ij}=\xi_{ij}^{(r)}$,
$\xi_{ij}^{(r)}=\hat{\alpha}_{ij}^{(r)}-\rho^{(r)}(\bm{c}_{ij})$.
Assume that the propagation field $\rho_{k}^{(r)}(\bm{z})$ in (\ref{eq:propagation model})
is third-order differentiable. Then, if $\phi_{k}^{(r)}$ is available,
we have the following results to characterize the construction error
$\xi_{ij}$ at a point $\bm{c}_{ij}$.
\begin{thm}
[Spectrum overlap for multiple sources]\label{thm:multiple sources}Let
$\hat{\alpha}_{ij}^{(r)}$ be the estimate obtained from the solution
to the integrated interpolation and \ac{btd} problem (\ref{eq:theta_ij least square})
under $\nu=0$ based on the regression model (\ref{eq:2nd order LPR model}).
The variance $\mathcal{E}_{\text{m}}=\mathbb{V}\{\xi_{ij}\}$ of the
interpolation error $\xi_{ij}$ for the integrated interpolation and
\ac{btd} approach is given by
\begin{align}
\mathcal{E}_{\text{m}} & =\Biggl(\varpi_{\eta}(\eta)\sigma_{\eta}^{2}+\varpi_{\epsilon}(\eta)\sigma_{\epsilon}^{2}\Biggl)C(\{\bm{z}_{m},b\})\label{eq:multiple sources}
\end{align}
where $\varpi_{\eta}(\eta)=(2-10\eta^{2}+10\eta-2\eta^{3})/(K(1-3\eta^{2}+2\eta)^{2})$
and $\varpi_{\epsilon}(\eta)=2(-\eta-1)/(K(3\eta^{2}-2\eta-1)).$
\end{thm}
In addition, $\varpi_{\eta}(\eta),\varpi_{\epsilon}(\eta)\geq\frac{2}{K}$
with lower bound achieved when $\eta=0$.
\begin{proof}
See Appendix \ref{sec:R=00003D2-sources}.
\end{proof}
It is found that when $\eta=0$, i.e., there is no overlap of frequency
bands between each source, we have the coefficients $\varpi_{\eta}(\eta)$
and $\varpi_{\epsilon}(\eta)$ to be $2/K$, which is consistent with
the results in Theorem \ref{thm:tensor-guided interpolation (full)}
for a signal case, where each source occupies $K/2$ frequency bands.
Increasing the overlapping ratio $\eta$ also increases the coefficients
$\varpi_{\eta}(\eta)$ and $\varpi_{\epsilon}(\eta)$ in (\ref{eq:multiple sources}),
leading to a worse performance. When $\eta$ reaches 1, we will mathematically
have the matrix in (\ref{eq: inverse broken}) to be singular, and
the estimation error can be arbitrarily large. In this case, one cannot
separate the two sources.

\subsection{Low-rank Regularization\label{subsec:Low-rankness-Constraint}}

The nuclear norm regularization in the third term of (\ref{eq: BTD interp regularized})
provides two benefits: First, it allows sparse interpolation in the
first term $l(f)$, where $\mathcal{I}$ in (\ref{eq:l(f)}) may only
contain a small set of grids to interpolate, and second, it imposes
low-rank of the propagation field represented by the matrix $\bm{S}_{r}$.
Note that while the interpolation $l(f)$ in (\ref{eq: BTD interp regularized})
exploits the locally spatial correlation of the propagation field,
the low-rank regularization tries to exploit the \emph{global} structure,
where the signal strength may decrease in all directions following
roughly the {\em same} law as distance increases. 

To numerically study the performance gain due to low-rank regularization,
we select two common forms of propagation model. One is under exponential
scale (EXP-model) $g(d)=\alpha\text{exp}(-d^{\beta})$, with $\alpha=1$,
$\beta=1.5$, $h=0.1,$ another is under log-scale (LOG-model) $g(d)=\alpha-\beta\times\text{log}_{10}d$,
with $\alpha=18$, $\beta=5$, $h=0.1$, where $d=\sqrt{x^{2}+y^{2}+h^{2}}$
represents the distance from the source at the origin, $(x,y)$ is
the coordinate of the grid cell, and $h$ is the height. We choose
the number of sources $R=1$, dimension $N=N_{1}=N_{2}=31$, the number
of frequency bands $K=20$, and the number of measurements $M=\rho N^{2}$
with sampling rate $\rho=5\text{\%}-10\text{\%}$. The performance
criterion \ac{nmse} is the same as in Section \ref{sec:Numerical-Results}.

Fig.~\ref{fig:low-rankness constraint} compares two schemes, one
with the low-rank regularization where $\mu=0.01$ in (\ref{eq:svt}),
and the other without the low-rank regularization where $\mu=0$.
We choose $\nu=10^{-4}$ in (\ref{eq:theta_ij least square}) for
both the two schemes. Fig.~\ref{fig:low-rankness constraint} shows
that the low-rank regularization can enhance the accuracy in the reconstruction
by more than $10$\% for both propagation models compared to the case
without regularization under the sparse observations. It is observed
that the proposed scheme with low-rank regularization outperforms
a interpolation scheme without the low-rank regularization.
\begin{figure}
\includegraphics{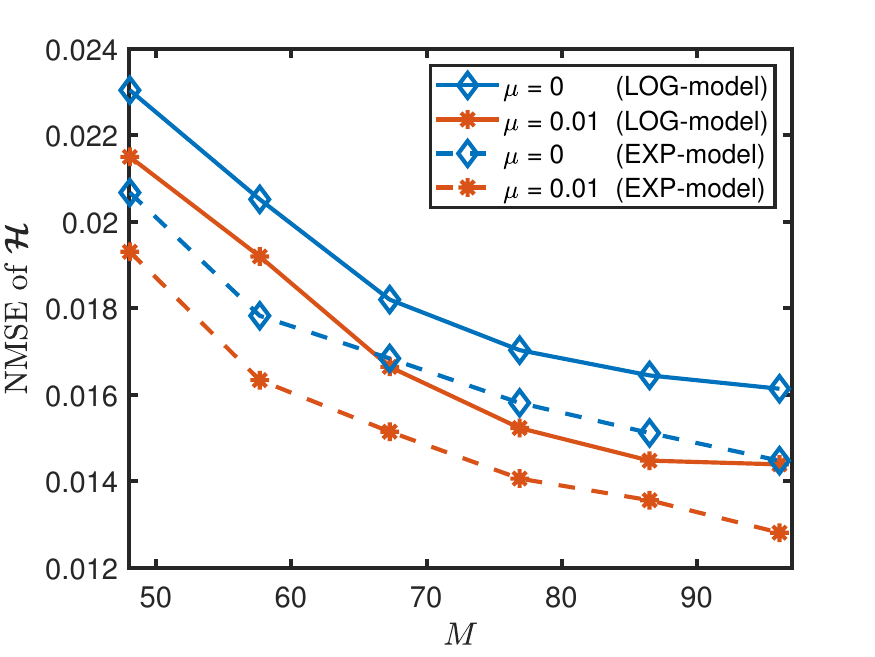}

\caption{\label{fig:low-rankness constraint}Comparison of the integrated interpolation
and \ac{btd} approach with and without the nuclear nrom low-rank
regularization, for parameters $\mu=0.01$ and $\mu=0$ respectively.}
\end{figure}

\subsection{Identifiability of Multiple Sources with Overlapping Spectrum}

One important feature of the proposed method is the capability of
identifying multiple sources possibly with overlapping spectrum. This
feature was also exploited by the related works based on tensor decomposition
in the literature \cite{ZhaFuWan:J20,CheWanZha:J23}. However, the
proposed scheme has the advantage that it does not require a priori
knowledge of the rank of $\bm{S}_{r}$. By contrast, existing schemes
based on the \ac{btd} model require the rank information as $\bm{S}_{r}$
needs to be decomposed into $\bm{S}_{r}=\bm{A}_{r}\bm{B}_{r}^{T}$\cite{ZhaFuWan:J20,CheWanZha:J23}.
Estimating the rank parameter $L$ for the existing schemes can be
challenging, because a small $L$ may lose some accuracy, whereas,
a large $L$ may lead to identifiability issue.

To numerically analyze the identifiability for multiple sources with
overlapping spectrum, we choose the overlapping ratio $\eta=0-0.6$,
then, the number of overlapped frequency bands is $\eta K$ with a
total $K(1+\eta)/2$ frequency bands observed of each source. We set
$\phi_{k}^{(r)}=1$, and then normalize $\bm{\phi}^{(r)}$ for comparison.
We choose $R=2$, $K=20$, $\rho=5\text{\%}$, other settings are
the same as in Section \ref{sec:Numerical-Results}. Since the \ac{tps}
method reconstructs the whole tensor without separating the source,
we only compare the proposed method with the TPS-BTD baseline in
Section \ref{sec:Numerical-Results}.

Fig.~\ref{fig:Performance-under-different-eta}(a) demonstrates that
the proposed method achieves more than double the improvement in spectrum
$\bm{\Phi}$ reconstruction compared to the TPS-BTD baseline. Furthermore,
as demonstrated in Fig.~\ref{fig:Performance-under-different-eta}(b),
it also provides an enhancement in power spectrum map $\bm{\mathcal{H}}$
reconstruction exceeding 10\%. This highlights its ability to identify
multiple sources, even with overlapping spectra, compared to the baseline
method. The TPS-BTD method underperforms primarily because selecting
an appropriate rank $L$ is challenging. Fig.~\ref{fig:Reconstruction of  propagation field}
displays a clear visualization of the propagation field $\bm{S}_{r}$
under $\eta=0.2$ where there is also an accuracy improvement of more
than $10$\% in reconstructing $\bm{S}_{r}$ compared to the baseline
method.
\begin{figure}
\subfigure[Reconstruction of  $\bm{\Phi}$]{\includegraphics[width=0.5\columnwidth]{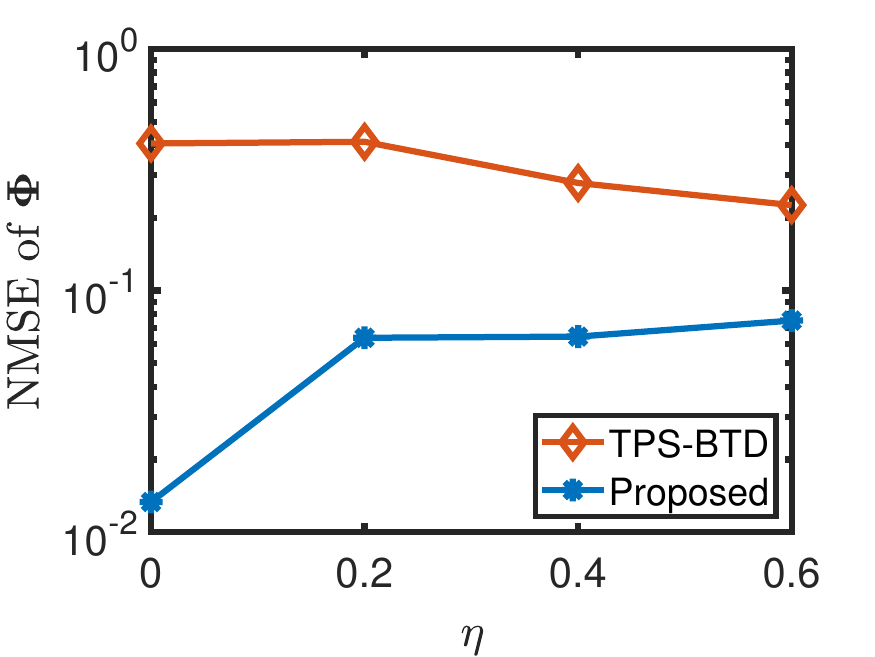}}\subfigure[Reconstruction of tensor $\bm{\mathcal{H}}$]{\includegraphics[width=0.5\columnwidth]{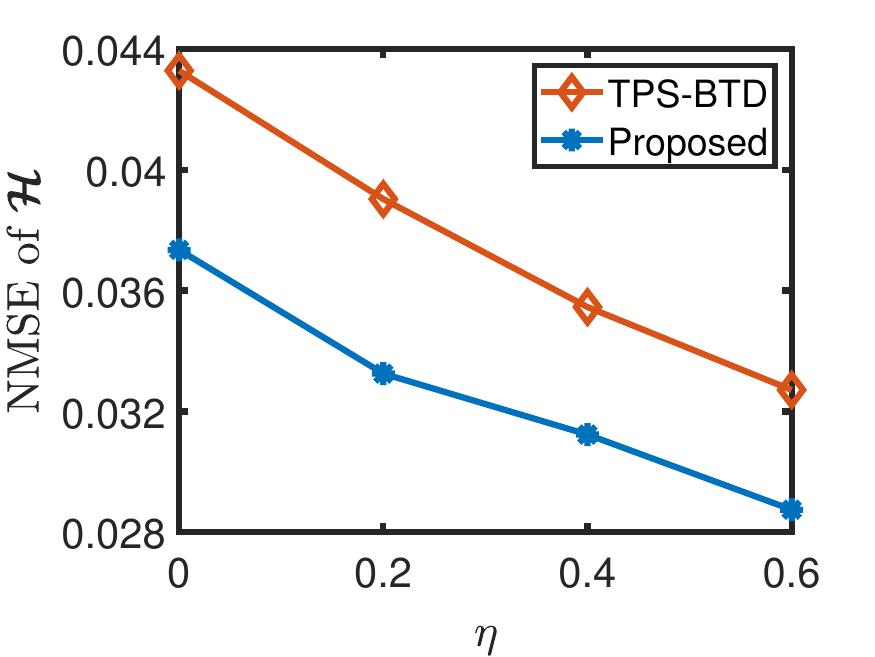}}\caption{\label{fig:Performance-under-different-eta}NMSE of the reconstructed
source spectrum $\bm{\Phi}$ and the power spectrum map $\bm{\mathcal{H}}$
versus the spectrum overlapping ratio $\eta$. The proposed scheme
works much better in separating the two sources.}
\end{figure}
\begin{figure}
\begin{centering}
\includegraphics[width=1\columnwidth]{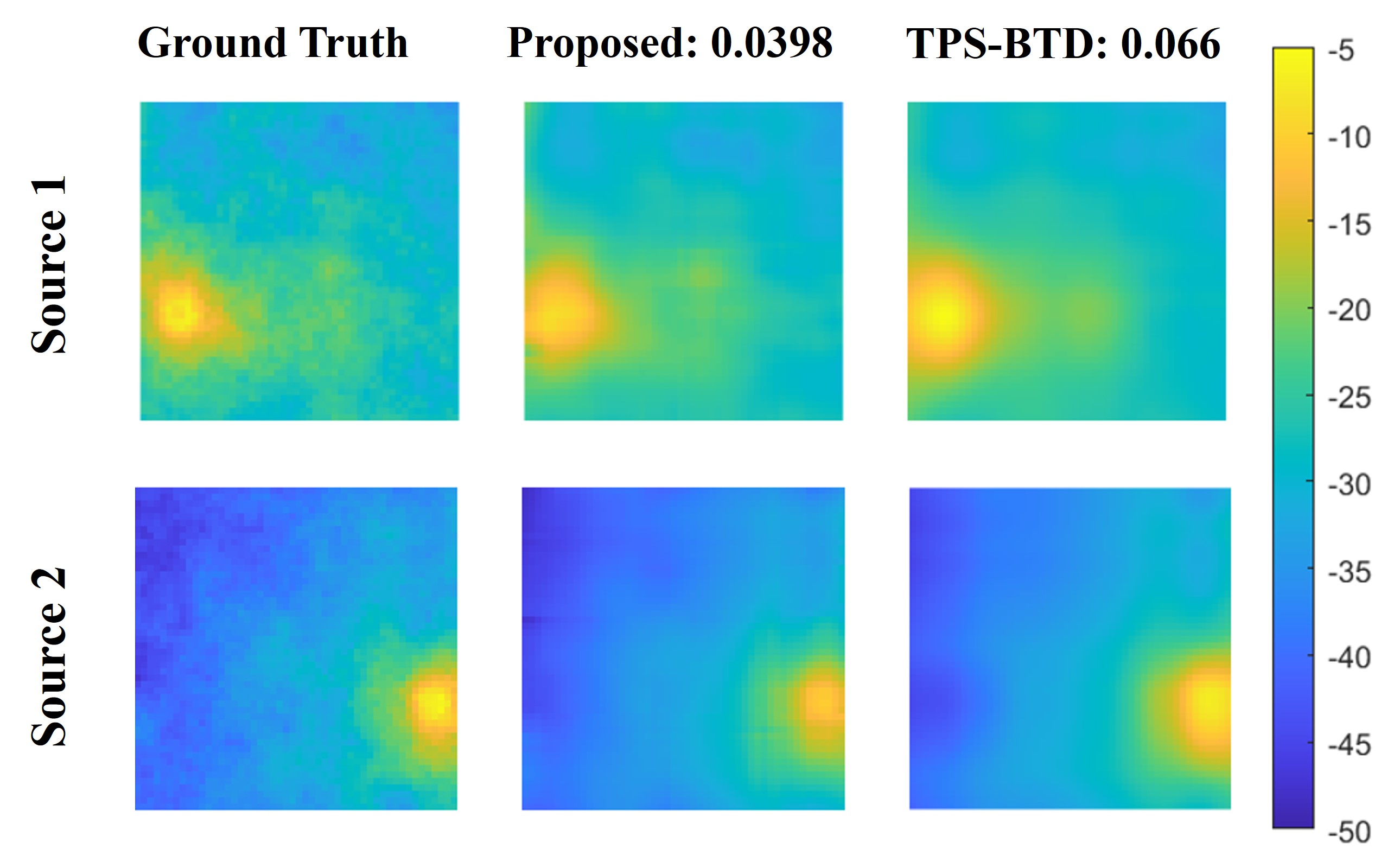}
\par\end{centering}
\caption{\label{fig:Reconstruction of  propagation field}Reconstruction of
the propagation field $\bm{S}_{r}$ of each source with the \ac{nmse}.
The sources constructed by TPS-BTD tends to be over dispersed, whereas,
the reconstruction of the proposed scheme appears to be more accurate.}
\end{figure}

\section{Numerical Results\label{sec:Numerical-Results}}

We adopt model (\ref{eq: measurement model}) to simulate the power
spectrum map in an $L\times L$ area for $L=50$ meters, where $g_{r}(d)=P_{r}(C_{0}/d)^{2}$
follows Friis transmission equation, $d$ represents the distance
from the source. We choose the parameter $P_{r}=1$W, $C_{0}=2$ for
illustrative purpose. Other values are broadly similar. The number
of sources is $R=2$. The power spectrum $\phi_{k}^{(r)}$ is generated
by $\phi_{k}^{(r)}=\sum_{i=1}^{2}a_{i}^{(r)}\text{sinc}^{2}(k-f_{i}^{(r)}/b_{i}^{(r)}),$
where $a_{i}^{(r)}\sim\mathcal{U}(0.5,2)$, $f_{i}^{(r)}\in\{1,\cdots,K\}$
is the center of the $i$-th square sinc function, $b_{i}^{\ensuremath{(r)}}\sim\mathcal{U}(2,4)$.
The sensors are distributed uniformly at random in the area to collect
the signal power $\gamma_{m}^{(k)}$. The shadowing component in log-scale
$\mbox{log}_{10}\zeta$ is modeled using a Gaussian process with zero
mean and auto-correlation function $\mathbb{E}\{\mbox{log}_{10}\zeta(\bm{z}_{i})\mbox{log}_{10}\zeta(\bm{z}_{j})\}=\sigma_{\text{s}}^{2}\mbox{exp}(-||\bm{z}_{i}-\bm{z}_{j}||_{2}/d_{\text{c}})$,
in which correlation distance $d_{\text{c}}=30$ meters, shadowing
variance $\sigma_{s}=4$. The \ac{snr} is defines as $\text{SNR}=\gamma_{m}^{(k)}/\tilde{\epsilon}_{k}$
where the $\tilde{\epsilon}_{k}$ follows Gaussian distribution $\mathcal{N}(0,\sigma^{2})$
and $\sigma^{2}$ is chosen to make the \ac{snr} $20$dB. We select
the parameter $b$ in the kernel function to ensure that it contains
at least $M_{0}=14$ sensors.

We employ the \ac{nmse} of the reconstructed power spectrum map for
performance evaluation. Let the NMSE of tensor $\bm{\mathcal{H}}$
be $||\hat{\bm{\mathcal{H}}}-\bm{\mathcal{H}}||_{F}^{2}/\|\bm{\mathcal{H}}\|_{F}^{2}$,
the NMSE of $\bm{\Phi}$ and $\bm{S}_{r}$ is of the same form. The
performance is compared with the following baselines that are recently
developed or adopted in related literature. Baseline~1: Thin plate
spline (TPS) \cite{JuaGonGeo:J11}. Baseline~2: low-rank tensor
completion (LRTC) \cite{LiuMusWon:J13}. Baseline~3: TPS-BTD\cite{CheWanZha:J23},
we first perform TPS, then, the uncertainty is derived and imposed
as the restriction on the BTD method.
\begin{figure}
\includegraphics[width=1\columnwidth]{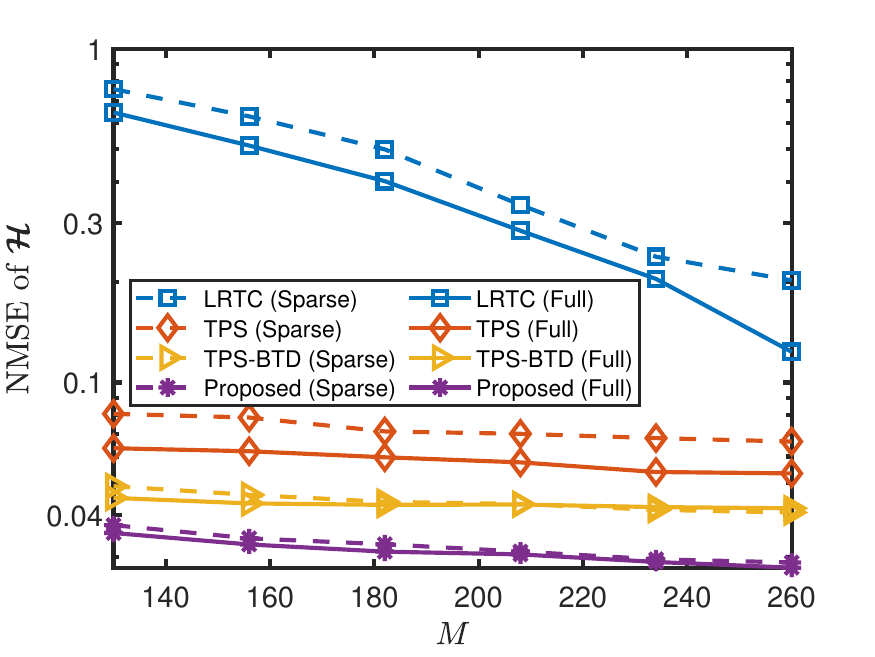}\caption{\label{fig:M full K}Reconstruction NMSE of $\bm{\mathcal{H}}$ when
observing the full spectrum or when randomly and sparsely observing
part of the spectrum.}
\end{figure}

\subsection{The Influence of Number of Measurements $M$}

We quantify the power spectrum map reconstruction performance of the
proposed method under different number of measurements $M=130-260$
which is of sampling rate $\rho=5\text{\%}-10\text{\%}$ under fixed
resolution $N=N_{1}=N_{2}=51$. We choose the number of frequency
bands being $|\Omega_{m}|=K=20$ for the full observation case and
$|\Omega_{m}|=K/2=10$ for the sparse observation case, separately.
For the sparse observation case, we randomly select the observed frequency
bands for the sensors.

Fig.~\ref{fig:M full K} shows the performance of power spectrum
reconstruction under both full and sparse observation scenarios. It
demonstrates that the proposed method is robust when each sensor only
observes random portion of the spectrum. The proposed scheme achieves
over a $20$\% improvement in reconstruction accuracy in both cases,
which translates to roughly $50$\% reduction of the measurements
to achieve the same \ac{nmse} performance. A larger $M$ will contribute
to a performance gain of the proposed method. The inferior performance
of LRTC stems from off-grid sparse observations. Similarly, TPS underperforms
because it fails to leverage the correlation property in the frequency
domain. TPS-BTD, which relies on TPS interpolation, also falls short
as it does not utilize the correlation property effectively.

When the number of measurements $M$ increase, the performance of
the proposed method under the sparse observation case approaches the
full observation case, revealing that exploit a tensor structure can
save the number of measurements to attain a similar accuracy.
\begin{figure}
\includegraphics{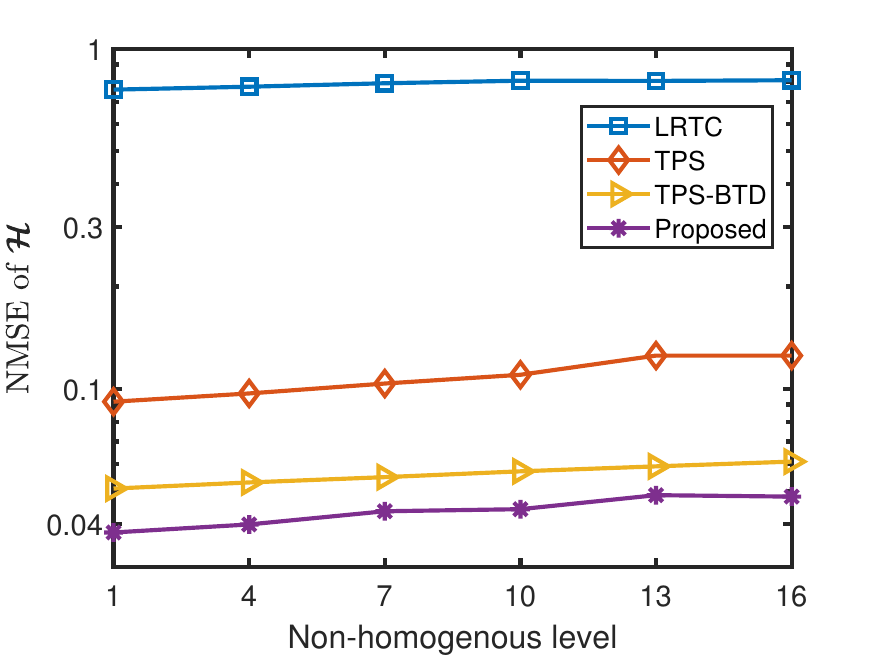}\caption{\label{fig:The-performance-under non-homogeneous}Reconstruction NMSE
of $\bm{\mathcal{H}}$ under non-homogeneous observation of spectrum.}
\end{figure}

\subsection{The Influence of Non-homogeneous Spectrum Observation}

We investigate the influence of the non-homogeneous observation of
spectrum on the performance under the sparse spectrum observation.
In the homogeneous observation, for each sensor location $\bm{z}_{m}$,
we randomly select the $|\Omega_{m}|=K/2$ frequency bands from the
whole bands. As a result, we will have a similar number of sensor
measurements for each frequency band. In the non-homogeneous case,
we still make that each frequency band $k$ have a similar number
of sensor measurements. But the difference lies in that for a set
of sensor locations, they are with higher probability to collect the
first $K/2$ bands, and for another set of sensor locations, they
are with higher probability to collect the last $K/2$ bands.

To realize this, we separate the sensors locations $\{\bm{z}_{m}\}$
to two sets $\mathcal{Z}_{1}$, $\mathcal{Z}_{2}$. Then, we choose
the weighted sampling without replacement method to choose $|\Omega_{m}|=10$
bands for each sensor from $K=20$ bands. For the locations in $\mathcal{Z}_{1}$,
we choose the weight to be $w_{k}^{(1)}=1$ for $1\leq k\leq10$,
and $w_{k}^{(1)}=C$ for $11\leq k\leq20$. For the locations in $\mathcal{Z}_{2}$,
we choose the weight to be $w_{k}^{(2)}=C$ for $1\leq k\leq10$,
and $w_{k}^{(2)}=1$ for $11\leq k\leq20$. Therefore, the weight
$C$ serves as a measure of the level of non-homogeneity. For $C=1$,
it is homogeneous with identical weights. As $C$ increases, the observations
of the spectrum become increasingly non-homogeneous.

Fig.~\ref{fig:The-performance-under non-homogeneous} illustrates
that the proposed method surpasses the baseline methods and is stable
under inhomogeneity of the measurement topology, achieving an improvement
of around $20\text{\%}$ in reconstruction accuracy. The proposed
method effectively handles extremely non-homogeneous spectrum cases
and still maintains a significant improvement over the baselines.
\begin{figure}
\subfigure[]{\includegraphics[width=0.5\columnwidth]{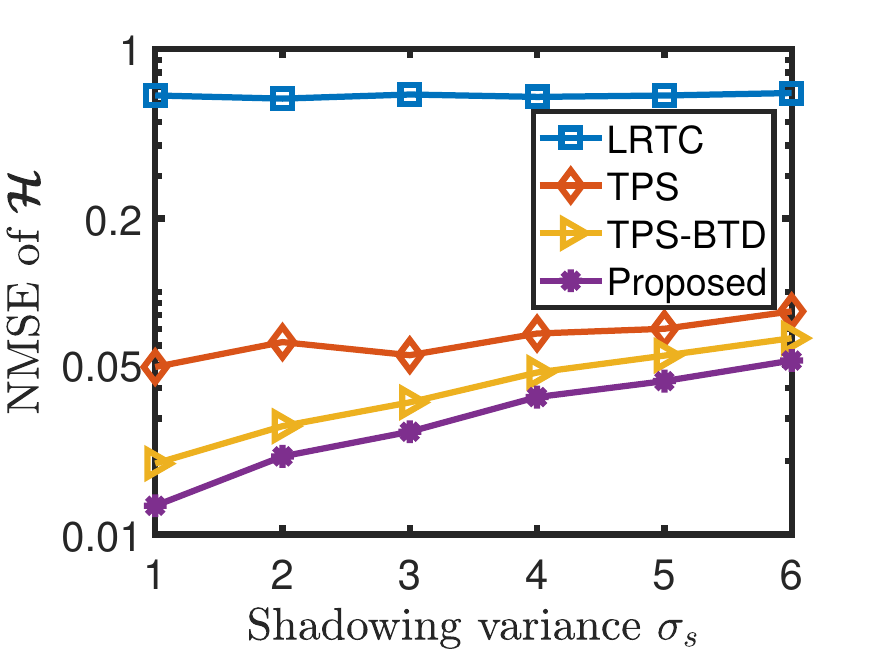}}\subfigure[]{\includegraphics[width=0.5\columnwidth]{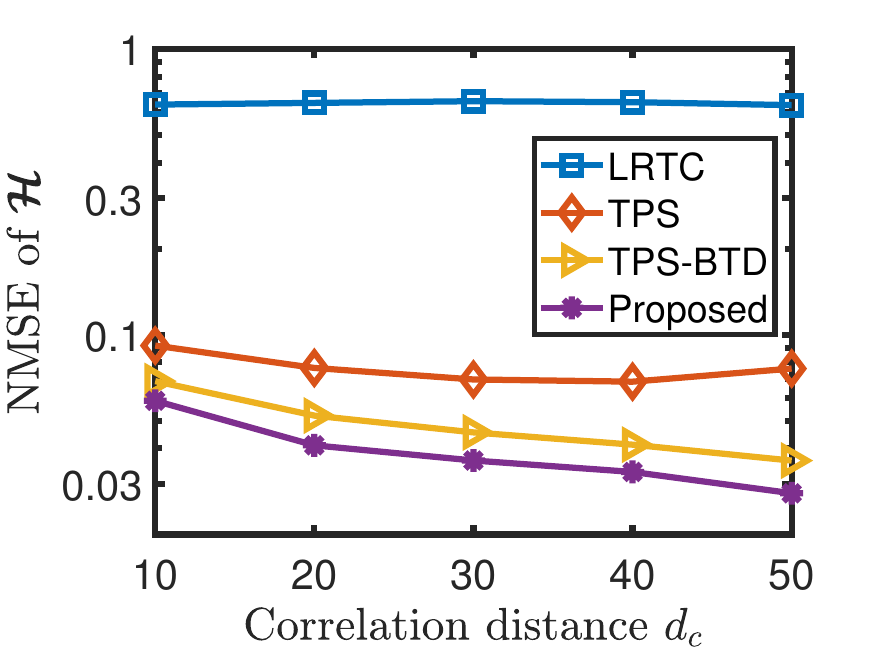}}\caption{\label{fig:NMSE versus sigma-1}Reconstruction NMSE of $\bm{\mathcal{H}}$
under various shadowing variance $\sigma_{s}$ and correlation distance
$d_{c}$.}
\end{figure}

\subsection{The Influence of Different Shadowing Parameters}

We quantify the reconstruction performance of the proposed method
under different shadowing variance $\sigma_{s}=1-6$ and different
correlation distance $d_{c}=10-50.$ We choose the full observation
case $K=20$, and $M=130$ which corresponding to $\rho=5$\%.

We first choose correlation distance $d_{c}=30$ and varying $\sigma_{s}$
from $1$ to $6$. Fig.~\ref{fig:NMSE versus sigma-1}(a) demonstrates
the performance of the proposed method outperforms the baseline methods
with an $18\text{\%}-31\text{\%}$ improvement under the shadowing
variances $\sigma_{s}=1-6$. The increasing of the shadowing variance
will cause a degradation in the performance of all the methods.

Then, we choose $\sigma_{s}=4$ and varying $d_{c}$ from $10$ to
$50$. Fig.~\ref{fig:NMSE versus sigma-1}(b) demonstrates that the
proposed method surpasses the baseline methods, showing an improvement
of $14\text{\%}-22\text{\%}$ across the correlation distances $d_{c}=10-50$.
Additionally, a larger correlation distance contributes to the performance
gain.

\subsection{The Influence of Off-grid Measurements}

\begin{figure}
\includegraphics{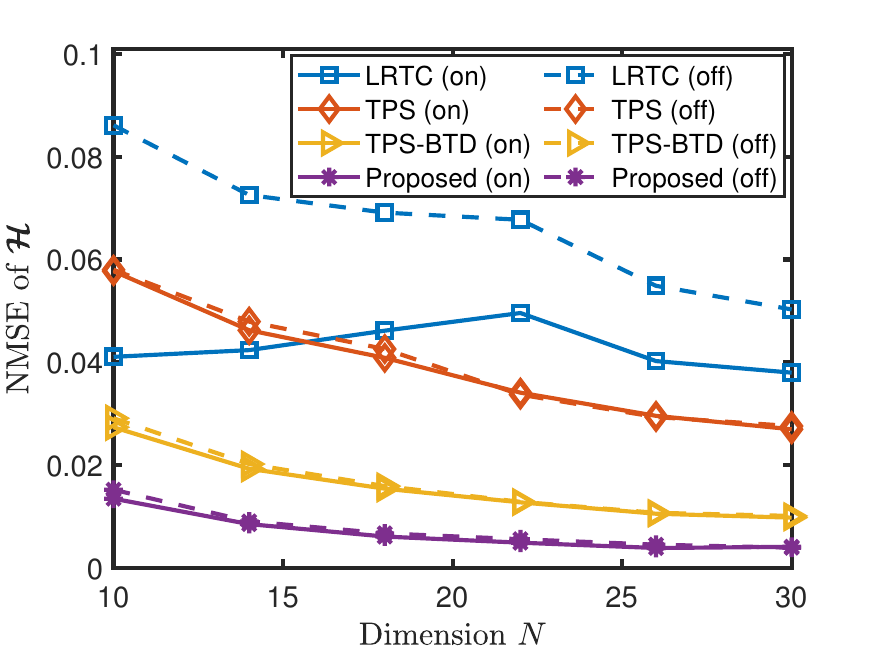}\caption{\label{fig:Compare-between-off-grid}Reconstruction NMSE of $\bm{\mathcal{H}}$
under off-grid and on-grid measurements.}
\end{figure}
For tensor completion, an off-grid measurements can cause large error
and one generally prefers a measurement collecting from the grid center.
The proposed method can naturally deal with the off-grid measurements
case. Here, we numerically study the influence of off-grid measurements
issue through comparing it with on-grid measurements to showcase the
proposed method.

The simulation is performed under $N=10-30$, $K=20$, $M=C_{1}N\text{log}^{2}(N)$
with $C_{1}=2$. Fig.~\ref{fig:Compare-between-off-grid} shows that
the proposed method outperforms baselines in over $20\text{\%}$ improvement
in reconstruction accuracy. In addition, the proposed method is stable
under off-grid measurements which shows a similar performance between
the on-grid and off-grid measurements. This is because the proposed
method has an interpolation as a intermediate step, it can deal with
the off-grid measurements. The performance between the off-grid and
on-grid case approaches at large $N$. The difference between on-grid
and off-grid for LRTC is large because the tensor completion method
is sensitive to the noise.

\section{Conclusion\label{sec:Conclusion}}

This paper developed a novel integrated interpolation and BTD approach
for reconstructing power spectrum maps. This integrated approach addressed
the spatial correlation of power spectrum maps, captured their low-rank
characteristics, and reduced errors from off-grid measurements. An
algorithm based on alternating least squares and singular value thresholding
was developed. Theoretical analysis showed that the incorporation
of the BTD structure improved interpolation. Extensive simulations
were conducted, demonstrating that the proposed method achieved over
$10$\% improvement in reconstruction accuracy across various scenarios,
such as different numbers of sensors, measurement topology homogeneity,
shadowing parameters, and off-grid scenarios.

% ----

\appendices{}

% ----

\section{Proof of Theorem \label{sec:Proof-of-Theorem k bands}\ref{thm:tensor-guided interpolation (full)}}

Under $\nu=0$, the solution for the weighted least-squares problem
(\ref{eq:theta_ij least square}) is
\begin{equation}
\hat{\bm{\Theta}}_{ij}=(\bm{\bm{\Phi}}\otimes\bm{X}_{ij}\bm{W}_{ij}^{2}\bm{\bm{\Phi}}^{\text{T}}\otimes\bm{X}_{ij}^{\text{T}})^{-1}\bm{\bm{\Phi}}\otimes\bm{X}_{ij}\bm{W}_{ij}^{2}\text{vec}(\bm{\Gamma}).\label{eq:result thetaij}
\end{equation}
Denote $\bm{\gamma}_{k}=[\gamma_{1}^{(k)},\gamma_{2}^{(k)},\cdots,\gamma_{M}^{(k)}]^{\text{T}}\in\mathbb{R}^{M}$
and recall that $\bm{\Gamma}(m,k)=\gamma_{m}^{(k)}$. From (\ref{eq: measurement model}),
we have,
\begin{align}
 & \text{vec}(\bm{\Gamma})\nonumber \\
 & =[\bm{\gamma}_{1}^{\text{T}},\bm{\gamma}_{2}^{\text{T}},\cdots,\bm{\gamma}_{K}^{\text{T}}]^{\text{T}}\nonumber \\
 & =[\rho(\bm{\bm{z}}_{1})\phi_{1}\cdots\rho(\bm{\bm{z}}_{M})\phi_{1},\cdots,\rho(\bm{\bm{z}}_{1})\phi_{K}\cdots\rho(\bm{\bm{z}}_{M})\phi_{K}]^{\text{T}}+\tilde{\bm{\epsilon}}\nonumber \\
 & =\tilde{\bm{\rho}}+\tilde{\bm{\epsilon}}\label{eq:gamma}
\end{align}
where $\tilde{\bm{\epsilon}}=[\tilde{\epsilon}_{1},\cdots,\tilde{\epsilon}_{1},\cdots,\tilde{\epsilon}_{K},\cdots,\tilde{\epsilon}_{K}]^{\text{T}}\in\mathbb{R}^{MK\times1}$.

Applying third order Taylor's expansion to $\rho(\bm{z}_{m})$ at
the neighborhood of $\bm{c}_{ij}$, we have
\begin{align}
\rho(\bm{z}_{m}) & =\rho(\bm{c}_{ij})+(\nabla\rho(\bm{c}_{ij}))^{\text{T}}(\bm{c}_{ij})(\bm{z}_{m}-\bm{c}_{ij})\nonumber \\
 & \qquad+\frac{1}{2}(\bm{z}_{m}-\bm{c}_{ij})^{\text{T}}\nabla^{2}\rho(\bm{c}_{ij})(\bm{z}_{m}-\bm{c}_{ij})\nonumber \\
 & \qquad\quad+P_{3}(\bm{z}_{m}-\bm{c}_{ij})+o(\|\bm{z}_{m}-\bm{c}_{ij}\|^{3})\label{eq: third order rho}
\end{align}
where $\nabla\rho(\bm{c}_{ij})$ is the first order derivative of
$\rho(\bm{z})$ evaluated at point $\bm{c}_{ij}$, $\nabla^{2}\rho(\bm{c}_{ij})$
is the Hessian matrix of $\rho(\bm{z})$ evaluated at point $\bm{c}_{ij}$,
$o(x)$ means $\text{lim}_{x\to0}o(x)/x\to0$, $P(\bm{z}_{m}-\bm{c}_{ij})=\sum_{k_{1}+k_{2}=3}\frac{\nabla^{(k_{1}+k_{2})}\rho(\bm{c}_{ij})}{k_{1}!k_{2}!}(\bm{z}_{m}(1)-\bm{c}_{ij}(1))^{k_{1}}(\bm{z}_{m}(2)-\bm{c}_{ij}(2))^{k_{2}}$
represents third order component, $\bm{z}_{m}(n)$ represent the $n$th
element in $\bm{z}_{m}$.

Denote $\bm{P}_{ij}=[P(\bm{z}_{1}-\bm{c}_{ij}),\cdots,P(\bm{z}_{M}-\bm{c}_{ij})]^{\text{T}}$
and recall $\bm{\bm{\bm{\Phi}}=}[\phi_{1},\phi_{2},\cdots,\phi_{K}]$.
Then, the expression of $\tilde{\bm{\rho}}$ in (\ref{eq:gamma})
can be rearranged into the following matrix form
\begin{align}
\tilde{\bm{\rho}} & =\bm{\bm{\Phi}}^{\text{T}}\otimes\bm{X}_{ij}^{\text{T}}\left[\begin{array}{c}
\rho(\bm{c}_{ij})\\
\nabla\rho(\bm{c}_{ij})\\
\nabla^{2}\rho(\bm{c}_{ij})
\end{array}\right]+\bm{\bm{\Phi}}^{\text{T}}\otimes\bm{P}_{ij}+\bm{r}_{ij}\label{eq:matrix form first order}
\end{align}
where $\bm{r}_{ij}$ is a vector of the residual terms $o(||\bm{z}_{m}-\bm{c}_{ij}||^{3})$.

Since $\mathbb{E}\{\tilde{\epsilon}\}=0$, the expectation of $\hat{\bm{\Theta}}_{ij}$
in (\ref{eq:result thetaij}) can be written as
\begin{align*}
 & \mathbb{E}\{\hat{\bm{\Theta}}_{ij}\}\\
 & =\mathbb{E}\{(\bm{\bm{\Phi}}\otimes\bm{X}_{ij}\bm{W}_{ij}^{2}\bm{\bm{\Phi}}^{\text{T}}\otimes\bm{X}_{ij}^{\text{T}})^{-1}\bm{\bm{\Phi}}\otimes\bm{X}_{ij}\bm{W}_{ij}^{2}\text{vec}(\bm{\Gamma})\}\\
 & =(\bm{\bm{\Phi}}\otimes\bm{X}_{ij}\bm{W}_{ij}^{2}\bm{\bm{\Phi}}^{\text{T}}\otimes\bm{X}_{ij}^{\text{T}})^{-1}\bm{\bm{\Phi}}\otimes\bm{X}_{ij}\bm{W}_{ij}^{2}\tilde{\bm{\rho}}.
\end{align*}

The variance of $\hat{\bm{\Theta}}_{ij}$ can be derived as follows:
\begin{align}
\begin{aligned} & \mathbb{V}\{\hat{\bm{\Theta}}_{ij}\}\\
 & =\mathbb{E}\{(\hat{\bm{\Theta}}_{ij}-\mathbb{E}\{\hat{\bm{\Theta}}_{ij}\})^{2}\}\\
 & =\mathbb{E}\{((\bm{\bm{\Phi}}\otimes\bm{X}_{ij}\bm{W}_{ij}^{2}\bm{\bm{\Phi}}^{\text{T}}\otimes\bm{X}_{ij}^{\text{T}})^{-1}\bm{\bm{\Phi}}\otimes\bm{X}_{ij}\bm{W}_{ij}^{2}\text{vec}(\bm{\Gamma})\\
 & \qquad\qquad-(\bm{\bm{\Phi}}\otimes\bm{X}_{ij}\bm{W}_{ij}^{2}\bm{\bm{\Phi}}^{\text{T}}\otimes\bm{X}_{ij}^{\text{T}})^{-1}\bm{\bm{\Phi}}\otimes\bm{X}_{ij}\bm{W}_{ij}^{2}\tilde{\bm{\rho}})\}\\
 & =\mathbb{E}\{((\bm{\bm{\Phi}}\otimes\bm{X}_{ij}\bm{W}_{ij}^{2}\bm{\bm{\Phi}}^{\text{T}}\otimes\bm{X}_{ij}^{\text{T}})^{-1}\bm{\bm{\Phi}}\otimes\bm{X}_{ij}\bm{W}_{ij}^{2}\bm{\tilde{\epsilon}})^{2}\}\\
 & =\mathbb{E}\{(\bm{\bm{\Phi}}\otimes\bm{X}_{ij}\bm{W}_{ij}^{2}\bm{\bm{\Phi}}^{\text{T}}\otimes\bm{X}_{ij}^{\text{T}})^{-1}(\bm{\bm{\Phi}}\otimes\bm{X}_{ij}\bm{W}_{ij}^{2}\bm{\tilde{\epsilon}}\\
 & \qquad\times\bm{\tilde{\epsilon}}^{\text{T}}\bm{W}_{ij}^{2}\bm{\bm{\Phi}}^{\text{T}}\otimes\bm{X}_{ij}^{\text{T}})(\bm{\bm{\Phi}}\otimes\bm{X}_{ij}\bm{W}_{ij}^{2}\bm{\bm{\Phi}}^{\text{T}}\otimes\bm{X}_{ij}^{\text{T}})^{-1}\}.
\end{aligned}
\label{eq:V theta ij}
\end{align}

Since each sensor measures all the frequency bands, $\bm{W}_{ij}$,
defined in (\ref{eq:weights}), is identity matrix. Through simple
matrix multiplication, the matrix $\bm{\bm{\Phi}}\otimes\bm{X}_{ij}\bm{W}_{ij}^{2}\bm{\bm{\Phi}}^{\text{T}}\otimes\bm{X}_{ij}^{\text{T}}$
can be derived as follows:
\begin{equation}
(\bm{\bm{\Phi}}\otimes\bm{X}_{ij}\bm{W}_{ij}^{2}\bm{\bm{\Phi}}^{\text{T}}\otimes\bm{X}_{ij}^{\text{T}})^{-1}=\left(\sum_{k=1}^{K}\text{\ensuremath{\phi_{k}^{2}}}\bm{A}_{1}\right)^{-1}\label{eq:first term in LS}
\end{equation}
where
\[
\bm{A}_{1}=\left[\begin{array}{ccc}
a_{1} & \bm{b}_{1} & \bm{c}_{1}\\
\bm{b}_{1}^{\text{T}} & \bm{D}_{1} & \bm{E}_{1}\\
\bm{c}_{1}^{\text{T}} & \bm{E}_{1}^{\text{T}} & \bm{F}_{1}
\end{array}\right],
\]
with $a_{1}=\sum_{m}\kappa_{ij}(\bm{z}_{m}),$ $\bm{b}_{1}=\sum_{m}\kappa_{ij}(\bm{z}_{m})(\bm{z}_{m}-\bm{c}_{ij})^{\text{T}},$
$\bm{c}_{1}=\sum_{m}\kappa_{ij}(\bm{z}_{m})\text{vec}((\bm{z}_{m}-\bm{c}_{ij})(\bm{z}_{m}-\bm{c}_{ij})^{\text{T}})^{\text{T}},$
$\bm{D}_{1}=\sum_{m}(\bm{z}_{m}-\bm{c}_{ij})\kappa_{ij}(\bm{z}_{m})(\bm{z}_{m}-\bm{c}_{ij})^{\text{T}},$
$\bm{E}_{1}=\sum_{m}(\bm{z}_{m}-\bm{c}_{ij})\kappa_{ij}(\bm{z}_{m})\text{vec}((\bm{z}_{m}-\bm{c}_{ij})(\bm{z}_{m}-\bm{c}_{ij})^{\text{T}})^{\text{T}},$
$\bm{F}_{1}=\sum_{m}\text{vec}((\bm{z}_{m}-\bm{c}_{ij})(\bm{z}_{m}-\bm{c}_{ij})^{\text{T}})\kappa_{ij}(\bm{z}_{m})\text{vec}((\bm{z}_{m}-\bm{c}_{ij})(\bm{z}_{m}-\bm{c}_{ij})^{\text{T}})^{\text{T}}$.
Here, $\bm{A}_{1}$ is a topology matrix that only depends on the
sample locations $\bm{z}_{m}$, the grid position $\bm{c}_{ij}$,
and the kernel function $\kappa_{ij}$. In addition,
\begin{equation}
\bm{\bm{\Phi}}\otimes\bm{X}_{ij}\bm{W}_{ij}^{2}\bm{\tilde{\epsilon}}\bm{\tilde{\epsilon}}^{\text{T}}\bm{W}_{ij}^{2}\bm{\bm{\Phi}}^{\text{T}}\otimes\bm{X}_{ij}^{\text{T}}=\sum_{k=1}^{K}\text{\ensuremath{\phi_{k}^{2}}}(\text{\ensuremath{\phi_{k}^{2}}}\eta_{r,k}^{2}+\epsilon^{2})\bm{A}_{2}\label{eq:LS middle term}
\end{equation}
where
\[
\bm{A}_{2}=\left[\begin{array}{ccc}
a_{2} & \bm{b}_{2} & \bm{c}_{2}\\
\bm{b}_{2}^{\text{T}} & \bm{D}_{2} & \bm{E}_{2}\\
\bm{c}_{2}^{\text{T}} & \bm{E}_{2}^{\text{T}} & \bm{F}_{2}
\end{array}\right],
\]
with $a_{2}=\sum_{m}\kappa_{ij}^{2}(\bm{z}_{m}),$ $\bm{b}_{2}=\sum_{m}\kappa_{ij}^{2}(\bm{z}_{m})(\bm{z}_{m}-\bm{c}_{ij})^{\text{T}},$
$\bm{c}_{2}=\sum_{m}\kappa_{ij}^{2}(\bm{z}_{m})\text{vec}((\bm{z}_{m}-\bm{c}_{ij})(\bm{z}_{m}-\bm{c}_{ij})^{\text{T}})^{\text{T}},$
$\bm{D}_{2}=\sum_{m}(\bm{z}_{m}-\bm{c}_{ij})\kappa_{ij}^{2}(\bm{z}_{m})(\bm{z}_{m}-\bm{c}_{ij})^{\text{T}},$
$\bm{E}_{2}=\sum_{m}(\bm{z}_{m}-\bm{c}_{ij})\kappa_{ij}^{2}(\bm{z}_{m})\text{vec}((\bm{z}_{m}-\bm{c}_{ij})(\bm{z}_{m}-\bm{c}_{ij})^{\text{T}})^{\text{T}},$
$\bm{F}_{2}=\sum_{m}\text{vec}((\bm{z}_{m}-\bm{c}_{ij})(\bm{z}_{m}-\bm{c}_{ij})^{\text{T}})\kappa_{ij}^{2}(\bm{z}_{m})\text{vec}((\bm{z}_{m}-\bm{c}_{ij})(\bm{z}_{m}-\bm{c}_{ij})^{\text{T}})^{\text{T}}$.
Likewise, $\bm{A}_{2}$ is also a topology matrix that captures the
impact from the sample locations.

Thus, Equation (\ref{eq:V theta ij}) is further derived as follows:
\begin{align}
 & \mathbb{V}\{\hat{\bm{\Theta}}_{ij}\}\nonumber \\
 & =\mathbb{E}\bigg\{\left(\sum_{k=1}^{K}\text{\ensuremath{\phi_{k}^{2}}}\bm{A}_{1}\right)^{-1}\sum_{k=1}^{K}\text{\ensuremath{\phi_{k}^{2}}}(\text{\ensuremath{\phi_{k}^{2}}}\eta_{r,k}^{2}+\epsilon^{2})\bm{A}_{2}\nonumber \\
 & \qquad\qquad\qquad\qquad\qquad\qquad\times\left(\sum_{k=1}^{K}\text{\ensuremath{\phi_{k}^{2}}}\bm{A}_{1}\right)^{-1}\bigg\}\nonumber \\
 & =\frac{1}{(\sum_{k=1}^{K}\text{\ensuremath{\phi_{k}^{2}}})^{2}}\left[\bm{A}_{1}^{-1}\bm{A}_{2}\bm{A}_{1}^{-1}\right]_{(1,1)}\mathbb{E}\bigg\{\sum_{k=1}^{K}\text{\ensuremath{\phi_{k}^{2}}}(\text{\ensuremath{\phi_{k}^{2}}}\eta_{r,k}^{2}+\epsilon^{2})\bigg\}.\label{eq:inter V}
\end{align}

Since $\eta_{r,k}(\bm{z})\sim\mathcal{N}(0,\sigma_{\eta}^{2})$ and
$\epsilon\sim\mathcal{N}(0,\sigma_{\epsilon}^{2})$ in (\ref{eq:propagation model})
and (\ref{eq:measurement model-non selective}), we have
\begin{align}
 & \mathbb{E}\bigg\{\sum_{k=1}^{K}\text{\ensuremath{\phi_{k}^{2}}}(\text{\ensuremath{\phi_{k}^{2}}}\eta_{r,k}^{2}+\epsilon^{2})\bigg\}\nonumber \\
 & =\sum_{k=1}^{K}\text{\ensuremath{\phi_{k}^{2}}}\text{\ensuremath{\phi_{k}^{2}}}\mathbb{E}\left\{ \eta_{r,k}^{2}\right\} +\sum_{k=1}^{K}\text{\ensuremath{\phi_{k}^{2}}}\mathbb{E}\left\{ \epsilon^{2}\right\} \nonumber \\
 & =\sum_{k=1}^{K}\text{\ensuremath{\phi_{k}^{4}}}\sigma_{\eta}^{2}+\sum_{k=1}^{K}\text{\ensuremath{\phi_{k}^{2}}}\sigma_{\epsilon}^{2}.\label{eq:inter E}
\end{align}
Since $\bm{\Theta}_{ij}$ is deterministic, we have $\mathbb{V}\{\hat{\bm{\Theta}}_{ij}-\bm{\Theta}_{ij}\}=\mathbb{V}\{\hat{\bm{\Theta}}_{ij}\}$.
Under single source case, we have $[\bm{\Theta}_{ij}]_{(1,1)}=\alpha_{ij}$.
Thus, from (\ref{eq:inter V}) and (\ref{eq:inter E}), the variance
of interpolation error $\xi_{ij}$ becomes
\begin{align*}
\mathbb{V}\{\xi_{ij}\} & =\left[\mathbb{V}\{\hat{\bm{\Theta}}_{ij}-\bm{\Theta}_{ij}\}\right]_{(1,1)}\\
 & =\left(\frac{\sum_{k=1}^{K}\text{\ensuremath{\phi_{k}^{4}}}}{\left(\sum_{k=1}^{K}\text{\ensuremath{\phi_{k}^{2}}}\right)^{2}}\sigma_{\eta}^{2}+\frac{1}{\sum_{k=1}^{K}\text{\ensuremath{\phi_{k}^{2}}}}\sigma_{\epsilon}^{2}\right)C(\{\bm{z}_{m}\},b)
\end{align*}
where $C(\{\bm{z}_{m}\},b)=\left[\bm{A}_{1}^{-1}\bm{A}_{2}\bm{A}_{1}^{-1}\right]_{(1,1)}$
is a function of the sensor locations $\bm{z}_{m}$ and window size
$b$. This leads to the result in Theorem \ref{thm:tensor-guided interpolation (full)}.

\section{Proof of Proposition \ref{prop:Impact-from-the-spectrum}\label{sec:Proof-of-Proposition impact from the spectrum for tensor}}

For the error term related to frequency-selective fading $\sigma_{\eta}^{2}$,
we have
\[
\varpi_{\eta}(\bm{\bm{\Phi}})\triangleq\frac{\sum_{k=1}^{K}\text{\ensuremath{\phi_{k}^{4}}}}{\left(\sum_{k=1}^{K}\text{\ensuremath{\phi_{k}^{2}}}\right)^{2}}=\frac{1}{1+2\frac{\sum_{i\neq j}\text{\ensuremath{\phi_{i}^{2}}}\text{\ensuremath{\phi_{j}^{2}}}}{\sum_{k=1}^{K}\text{\ensuremath{\phi_{k}^{4}}}}}.
\]
Then, a larger value of $\sum_{i\neq j}\text{\ensuremath{\phi_{i}^{2}}}\text{\ensuremath{\phi_{j}^{2}}}/\sum_{k=1}^{K}\text{\ensuremath{\phi_{k}^{4}}}$
will contribute to a smaller value of $\varpi_{\eta}(\bm{\Phi})$,
vice versa. It can be verified easily that when there is equal-power
allocation, i.e., $\phi_{i}=\phi_{j}$, $\forall i\neq j$, $\sum_{i\neq j}\text{\ensuremath{\phi_{i}^{2}}}\text{\ensuremath{\phi_{j}^{2}}}/\sum_{k=1}^{K}\text{\ensuremath{\phi_{k}^{4}}}$
attains its largest value $(K-1)/2$, thus $\varpi_{\eta}(\bm{\Phi})$
attains its smallest value $\frac{1}{K}$. When there is a power boost
in $k$th frequency band, i.e., $\phi_{k}$ approaches $K$ and others
approaches $0$, $\sum_{i\neq j}\text{\ensuremath{\phi_{i}^{2}}}\text{\ensuremath{\phi_{j}^{2}}}/\sum_{k=1}^{K}\text{\ensuremath{\phi_{k}^{4}}}$
will attain its smallest value $0$, thus $\varpi_{\eta}(\bm{\bm{\Phi}})$
attains its largest value $1$.

For the error term related to measurement noise $\sigma_{\epsilon}^{2}$,
we have
\[
\varpi_{\epsilon}(\bm{\bm{\Phi}})\triangleq\frac{1}{\sum_{k=1}^{K}\text{\ensuremath{\phi_{k}^{2}}}}.
\]
It can be verified easily that it attains the largest value $1/K$
when there is equal-power allocation, i.e., $\phi_{i}=\phi_{j}$,
$\forall i\neq j$ and the smallest value $1/K^{2}$ when there is
a power boost in $k$th frequency band, i.e., $\phi_{k}$ approaches
$K$ and others approaches $0$. This leads to the result in Proposition
\ref{prop:Impact-from-the-spectrum}.

\section{Proof of Proposition \ref{pro:traditional LP}\label{sec:Proof-of-Proposition LP}}

The derivation follows the same approach for proving Theorem~\ref{thm:tensor-guided interpolation (full)},
resulting in the local polynomial interpolation error for the $k$th
frequency band as
\[
\mathbb{V}\{\xi_{ij}^{(k)}\}=\left(\sigma_{\eta}^{2}+\frac{\sigma_{\epsilon}^{2}}{\text{\ensuremath{\phi_{k}^{2}}}}\right)C(\{\bm{z}_{m}\},b)
\]
where the result coincides with (\ref{eq:thm-variance-tensor}) under
$K=1$. Note that since the frequency-selective fading $\eta_{r,k}$
in (\ref{eq:propagation model}) and the measurement noise $\epsilon$
in (\ref{eq:measurement model-non selective}) are assumed to be independent
across $k$, given the sensor topology $\bm{z}_{m}$ and $\bm{\phi}$,
the interpolation error $\xi_{ij}^{(k)}$ is also independent across
$k$. Therefore, the averaged interpolation error of all the $K$
frequency bands is
\[
\mathbb{V}\{\xi_{ij}\}=\frac{\sum_{k=1}^{K}\mathbb{V}\{\xi_{ij}^{(k)}\}}{K}=\left(\sigma_{\text{\ensuremath{\eta}}}^{2}+\sum_{k=1}^{K}\frac{1}{\text{\ensuremath{\phi_{k}^{2}K}}}\sigma_{\epsilon}^{2}\right)C(\{\bm{z}_{m}\},b).
\]

\section{Proof of Proposition \label{sec:Proof-of-Corollary sparse tensor guided}\ref{pro:difference of variance}}

We have the difference value of the variance of interpolation error
between the integrated interpolation and BTD approach and conventional
frequency-by-frequency interpolation as follows:
\begin{align}
 & \mathcal{E}_{\text{p}}-\mathcal{E}_{\text{t}}\nonumber \\
 & =\sigma_{\text{\ensuremath{\eta}}}^{2}+\sum_{k=1}^{K}\frac{1}{\text{\ensuremath{\phi_{k}^{2}K}}}\sigma_{\epsilon}^{2}-\frac{\sum_{k=1}^{K}\text{\ensuremath{\phi_{k}^{4}}}}{\left(\sum_{k=1}^{K}\text{\ensuremath{\phi_{k}^{2}}}\right)^{2}}\sigma_{\eta}^{2}-\frac{1}{\sum_{k=1}^{K}\text{\ensuremath{\phi_{k}^{2}}}}\sigma_{\epsilon}^{2}\nonumber \\
 & =\left(1-\frac{\sum_{k=1}^{K}\text{\ensuremath{\phi_{k}^{4}}}}{(\sum_{k=1}^{K}\text{\ensuremath{\phi_{k}^{2}}})^{2}}\right)\sigma_{\text{\ensuremath{\eta}}}^{2}\label{eq: part a}\\
 & \qquad+\left(\sum_{k=1}^{K}\frac{1}{\text{\ensuremath{\phi_{k}^{2}K}}}-\frac{1}{\sum_{k=1}^{K}\text{\ensuremath{\phi_{k}^{2}}}}\right)\sigma_{\epsilon}^{2}.\label{eq:part b}
\end{align}

For (\ref{eq: part a}), the coefficient of $\sigma_{\text{\ensuremath{\eta}}}^{2}$
can be derived as follows:
\[
1-\frac{\text{\ensuremath{\sum_{k=1}^{K}}\ensuremath{\ensuremath{\phi_{k}^{4}}}}}{(\sum_{k=1}^{K}\text{\ensuremath{\phi_{k}^{2}}})^{2}}=\frac{2\sum_{i\neq j}\text{\ensuremath{\phi_{i}^{2}}}\text{\ensuremath{\phi_{j}^{2}}}}{(\sum_{k=1}^{K}\text{\ensuremath{\phi_{k}^{2}}})^{2}}>0.
\]

For (\ref{eq:part b}), the coefficient of $\sigma_{\text{\ensuremath{\epsilon}}}^{2}$
can be derived as follows:
\begin{align*}
 & \left(\sum_{k=1}^{K}\frac{1}{\text{\ensuremath{\phi_{k}^{2}K}}}-\frac{1}{\sum_{k=1}^{K}\text{\ensuremath{\phi_{k}^{2}}}}\right)\\
 & =\frac{1}{\sum_{k=1}^{K}\text{\ensuremath{\phi_{k}^{2}}}}\left(\sum_{k=1}^{K}\left(\frac{\sum_{l\neq k}^{K}\text{\ensuremath{\phi_{l}^{2}}}}{\phi_{k}^{2}K}+\frac{1}{K}\right)-1\right)\\
 & =\frac{1}{\sum_{k=1}^{K}\text{\ensuremath{\phi_{k}^{2}}}}\sum_{k=1}^{K}\frac{\sum_{l\neq k}^{K}\text{\ensuremath{\phi_{l}^{2}}}}{\phi_{k}^{2}K}>0.
\end{align*}

Thus, $\mathcal{E}_{\text{p}}>\mathcal{E}_{\text{t}}$ and the difference
is
\[
\mathcal{E}_{\text{p}}-\mathcal{E}_{\text{t}}=\frac{2\sum_{i\neq j}\text{\ensuremath{\phi_{i}^{2}}}\text{\ensuremath{\phi_{j}^{2}}}}{(\sum_{k=1}^{K}\text{\ensuremath{\phi_{k}^{2}}})^{2}}\sigma_{\eta}^{2}+\frac{\sigma_{\epsilon}^{2}}{\sum_{k=1}^{K}\text{\ensuremath{\phi_{k}^{2}}}}\sum_{k=1}^{K}\frac{\sum_{l\neq k}^{K}\text{\ensuremath{\phi_{l}^{2}}}}{\phi_{k}^{2}K}.
\]

\section{Proof of Proposition \ref{prop:sparse spectrum interp}\label{sec:Proof-of-Proposition sparse spectrum}}

The derivation is the same as proving Theorem \ref{thm:tensor-guided interpolation (full)}
in Appendix \ref{sec:Proof-of-Theorem k bands}, except that, due
to the sparse observation, (\ref{eq:first term in LS}) becomes
\begin{equation}
(\bm{\bm{\Phi}}\otimes\bm{X}_{ij}\bm{W}_{ij}^{2}\bm{\bm{\Phi}}^{\text{T}}\otimes\bm{X}_{ij}^{\text{T}})^{-1}=\left(K'\bm{A}_{1}\right)^{-1}.\label{eq:sparse first term}
\end{equation}
To verify the above result, using Equation (\ref{eq:weights}) and
the definition of $\bm{\bm{\Phi}}$ and $\bm{X}_{ij}$, the first
element in the matrix $\bm{\bm{\Phi}}\otimes\bm{X}_{ij}\bm{W}_{ij}^{2}\bm{\bm{\Phi}}^{\text{T}}\otimes\bm{X}_{ij}^{\text{T}}$
is derived as follows:
\begin{align}
 & \left[\bm{\bm{\Phi}}\otimes\bm{X}_{ij}\bm{W}_{ij}^{2}\bm{\bm{\Phi}}^{\text{T}}\otimes\bm{X}_{ij}^{\text{T}}\right]_{(1,1)}\nonumber \\
 & =\sum_{m=1}^{M}\sum_{k=1}^{K}\phi_{k}^{2}\bm{\psi}(m,k)\kappa_{ij}(\bm{z}_{m})\nonumber \\
 & =\sum_{m=1}^{M}\sum_{k=1}^{K}\bm{\psi}(m,k)\kappa_{ij}(\bm{z}_{m})\nonumber \\
 & =K'\sum_{m=1}^{M}\kappa_{ij}(\bm{z}_{m})\label{eq:sparse equal-power-alloc}
\end{align}
which equals to $K'a_{1}$, where $a_{1}$ is defined in (\ref{eq:first term in LS})
and the last equation is because $|\Omega_{m}|=\sum_{k=1}^{K}\bm{\psi}(m,k)=K'$.

Similarly, one can easily show that all the other elements in the
matrix $\bm{\bm{\Phi}}\otimes\bm{X}_{ij}\bm{W}_{ij}^{2}\bm{\bm{\Phi}}^{\text{T}}\otimes\bm{X}_{ij}^{\text{T}}$
equals to the corresponding element in $\bm{A}_{1}$ scaled by $K'$.
Moreover, considering $\sum_{k=1}^{K}\psi(m,k)=K'$, Equation (\ref{eq:LS middle term})
becomes
\begin{equation}
\bm{\bm{\Phi}}\otimes\bm{X}_{ij}\bm{W}_{ij}^{2}\bm{\tilde{\epsilon}}\bm{\tilde{\epsilon}}^{\text{T}}\bm{W}_{ij}^{2}\bm{\bm{\Phi}}^{\text{T}}\otimes\bm{X}_{ij}^{\text{T}}=K'(\eta_{r,k}^{2}+\epsilon^{2})\bm{A}_{2}.\label{eq:sparse middle term}
\end{equation}
Then, the remaining part in Appendix \ref{sec:Proof-of-Theorem k bands}
can be directly applied based on the modified quantities (\ref{eq:sparse first term})
and (\ref{eq:sparse middle term}), leading to the following result
$\mathbb{V}\{\xi_{ij}\}=((\sigma_{\eta}^{2}+\sigma_{\epsilon}^{2})/K')C(\{\bm{z}_{m}\},b)$
in Proposition \ref{prop:sparse spectrum interp}.

\section{Proof of Theorem \ref{thm:multiple sources}\label{sec:R=00003D2-sources}}

Recall $\bm{\gamma}_{k}=[\gamma_{1}^{(k)},\gamma_{2}^{(k)},\cdots,\gamma_{M}^{(k)}]^{\text{T}}\in\mathbb{R}^{M}$
and $\bm{\Gamma}(m,k)=\gamma_{m}^{(k)}$. From (\ref{eq: measurement model}),
under $R=2$, we have $\gamma_{m}^{(k)}=\rho^{(1)}(\bm{\bm{z}}_{m})\phi_{k}^{(1)}+\tilde{\epsilon}_{k}+\rho^{(2)}(\bm{\bm{z}}_{m})\phi_{k}^{(2)}+\tilde{\epsilon}_{k}$.
Thus,
\begin{align}
\text{vec}(\bm{\Gamma}) & =[\bm{\gamma}_{1}^{\text{T}},\bm{\gamma}_{2}^{\text{T}},\cdots,\bm{\gamma}_{K}^{\text{T}}]^{\text{T}}\nonumber \\
 & =[\rho^{(1)}(\bm{\bm{z}}_{1})\phi_{1}^{(1)}+\rho^{(2)}(\bm{\bm{z}}_{1})\phi_{1}^{(2)},\cdots\nonumber \\
 & \qquad\rho^{(1)}(\bm{\bm{z}}_{M})\phi_{1}^{(1)}+\rho^{(2)}(\bm{\bm{z}}_{M})\phi_{1}^{(2)},\cdots\nonumber \\
 & \qquad\rho^{(1)}(\bm{\bm{z}}_{1})\phi_{K}^{(1)}+\rho^{(2)}(\bm{\bm{z}}_{1})\phi_{K}^{(2)},\cdots\nonumber \\
 & \qquad\rho^{(1)}(\bm{\bm{z}}_{M})\phi_{K}^{(1)}+\rho^{(2)}(\bm{\bm{z}}_{M})\phi_{K}^{(2)}]+\tilde{\bm{\epsilon}}\nonumber \\
 & =\tilde{\bm{\rho}}+\tilde{\bm{\epsilon}}\label{eq:=00005Cgamma/d_k-1-1-1}
\end{align}
where $\tilde{\bm{\epsilon}}=[\tilde{\epsilon}_{1},\cdots,\tilde{\epsilon}_{1},\cdots,\tilde{\epsilon}_{K},\cdots,\tilde{\epsilon}_{K}]^{\text{T}}\in\mathbb{R}^{MK\times1}$.

Similar to the derivation in (\ref{eq:matrix form first order}),
the expression of $\tilde{\bm{\rho}}$ in (\ref{eq:=00005Cgamma/d_k-1-1-1})
can be rearranged into the following matrix form
\begin{align}
\tilde{\bm{\rho}} & =\bm{\bm{\Phi}}^{\text{T}}\otimes\bm{X}_{ij}^{\text{T}}\bm{\rho}_{ij}+\bm{\bm{\Phi}}^{\text{T}}\otimes\bm{P}_{ij}+\bm{r}_{ij}\label{eq:matrix form first order-1-1}
\end{align}
where $\bm{P}_{ij}=[P^{(1)}(\bm{z}_{1}-\bm{c}_{ij}),\cdots,P^{(1)}(\bm{z}_{M}-\bm{c}_{ij}),P^{(2)}(\bm{z}_{1}-\bm{c}_{ij}),\cdots,P^{(2)}(\bm{z}_{M}-\bm{c}_{ij})]^{\text{T}}$,
$\bm{r}_{ij}$ is a vector of the residual terms $o(||\bm{z}_{m}-\bm{c}_{ij}||^{3})$,
and
\begin{align*}
\bm{\rho}_{ij} & =[\rho^{(1)}(\bm{c}_{ij}),\text{vec}(\nabla\rho^{(1)}(\bm{c}_{ij})),\text{vec}(\nabla^{2}\rho^{(1)}(\bm{c}_{ij})),\\
 & \qquad\qquad\rho^{(2)}(\bm{c}_{ij}),\text{vec}(\nabla\rho^{(2)}(\bm{c}_{ij})),\text{vec}(\nabla^{2}\rho^{(2)}(\bm{c}_{ij}))]^{\text{T}}.
\end{align*}

Similar as in Appendix \ref{sec:Proof-of-Theorem k bands}, the variance
of the interpolation error $\xi_{ij}$ is as follows:
\begin{align}
\mathbb{V}\{\xi_{ij}\} & =\bigg[(\bm{\bm{\Phi}}\otimes\bm{X}_{ij}\bm{W}_{ij}^{2}\bm{\bm{\Phi}}^{\text{T}}\otimes\bm{X}_{ij}^{\text{T}})^{-1}\nonumber \\
 & \qquad\times\mathbb{E}\left\{ (\bm{\bm{\Phi}}\otimes\bm{X}_{ij}\bm{W}_{ij}^{2}\bm{\tilde{\epsilon}}\bm{\tilde{\epsilon}}^{\text{T}}\bm{W}_{ij}^{2}\bm{\bm{\Phi}}^{\text{T}}\otimes\bm{X}_{ij}^{\text{T}})\right\} \nonumber \\
 & \qquad\times(\bm{\bm{\Phi}}\otimes\bm{X}_{ij}\bm{W}_{ij}^{2}\bm{\bm{\Phi}}^{\text{T}}\otimes\bm{X}_{ij}^{\text{T}})^{-1}\bigg]_{(1,1)}.\label{eq:multi sources variance}
\end{align}

Following a similar derivation in (\ref{eq:sparse first term}) and
(\ref{eq:sparse middle term}), we have
\begin{align}
 & (\bm{\bm{\Phi}}\otimes\bm{X}_{ij}\bm{W}_{ij}^{2}\bm{\bm{\Phi}}^{\text{T}}\otimes\bm{X}_{ij}^{\text{T}})^{-1}\nonumber \\
 & =\left[\begin{array}{cc}
K_{1}\bm{A}_{1} & K_{2}\bm{A}_{1}\\
K_{2}\bm{A}_{1} & K_{1}\bm{A}_{1}
\end{array}\right]^{-1}\label{eq: inverse broken}
\end{align}
where $K_{1}=(1+\eta)/2\times K$, $K_{2}=\eta K$, and
\begin{align*}
 & \mathbb{E}\left\{ (\bm{\bm{\Phi}}\otimes\bm{X}_{ij}\bm{W}_{ij}^{2}\bm{\tilde{\epsilon}}\bm{\tilde{\epsilon}}^{\text{T}}\bm{W}_{ij}^{2}\bm{\bm{\Phi}}^{\text{T}}\otimes\bm{X}_{ij}^{\text{T}})\right\} \\
 & =\left[\begin{array}{cc}
\left(M_{1}\sigma_{\eta}^{2}+K_{1}\sigma_{\epsilon}^{2}\right)\bm{A}_{2} & \left(M_{2}\sigma_{\eta}^{2}+K_{2}\sigma_{\epsilon}^{2}\right)\bm{A}_{2}\\
\left(M_{2}\sigma_{\eta}^{2}+K_{2}\sigma_{\epsilon}^{2}\right)\bm{A}_{2} & \left(M_{1}\sigma_{\eta}^{2}+K_{1}\sigma_{\epsilon}^{2}\right)\bm{A}_{2}
\end{array}\right]
\end{align*}
where $M_{1}=K_{1}+K_{2}$, $M_{2}=2K_{2}$.

Assume $\bm{A}_{1}$ is invertible and the Schur complement of $\left[\begin{array}{cc}
K_{1}\bm{A}_{1} & K_{2}\bm{A}_{1}\\
K_{2}\bm{A}_{1} & K_{1}\bm{A}_{1}
\end{array}\right]$ is also invertible, then we have
\begin{align*}
 & (\bm{\bm{\Phi}}\otimes\bm{X}_{ij}\bm{W}_{ij}^{2}\bm{\bm{\Phi}}^{\text{T}}\otimes\bm{X}_{ij}^{\text{T}})^{-1}\\
 & =\left[\begin{array}{cc}
\frac{K_{1}}{K_{1}^{2}-K_{2}^{2}}\bm{A}_{1}^{-1} & -\frac{K_{2}}{K_{1}^{2}-K_{2}^{2}}\bm{A}_{1}^{-1}\\
-\frac{K_{2}}{K_{1}^{2}-K_{2}^{2}}\bm{A}_{1}^{-1} & \frac{K_{1}}{K_{1}^{2}-K_{2}^{2}}\bm{A}_{1}^{-1}
\end{array}\right]\text{\ensuremath{\triangleq}}\bm{\mathrm{A}}.
\end{align*}
 Then, (\ref{eq:multi sources variance}) equals
\begin{align*}
\left[\bm{\mathrm{A}}\left[\begin{array}{cc}
(M_{1}\sigma_{\eta}^{2}+K_{1}\sigma_{\epsilon}^{2})\bm{A}_{2} & (M_{2}\sigma_{\eta}^{2}+K_{2}\sigma_{\epsilon}^{2})\bm{A}_{2}\\
(M_{2}\sigma_{\eta}^{2}+K_{2}\sigma_{\epsilon}^{2})\bm{A}_{2} & (M_{1}\sigma_{\eta}^{2}+K_{1}\sigma_{\epsilon}^{2})\bm{A}_{2}
\end{array}\right]\bm{\mathrm{A}}\right]_{(1,1)}.
\end{align*}
Since $\sigma_{\eta}^{2}$ and $\sigma_{\epsilon}^{2}$ are independent,
we can analyze $\mathbb{V}\{\xi_{ij}\}$ separately. Through simple
matrix operations, we have
\begin{align*}
 & \left[\bm{\mathrm{A}}\left[\begin{array}{cc}
K_{1}\sigma_{\epsilon}^{2}\bm{A}_{2} & K_{2}\sigma_{\epsilon}^{2}\bm{A}_{2}\\
K_{2}\sigma_{\epsilon}^{2}\bm{A}_{2} & K_{1}\sigma_{\epsilon}^{2}\bm{A}_{2}
\end{array}\right]\bm{\mathrm{A}}\right]_{(1,1)}\\
 & =\frac{-2\eta-2}{K\left(3\eta^{2}-2\eta-1\right)}\sigma_{\epsilon}^{2}C(\{\bm{z}_{m},b\})
\end{align*}
and
\begin{align*}
 & \left[\bm{\mathrm{A}}\left[\begin{array}{cc}
M_{1}\sigma_{\eta}^{2}\bm{A}_{2} & M_{2}\sigma_{\eta}^{2}\bm{A}_{2}\\
M_{2}\sigma_{\eta}^{2}\bm{A}_{2} & M_{1}\sigma_{\eta}^{2}\bm{A}_{2}
\end{array}\right]\bm{\mathrm{A}}\right]_{(1,1)}\\
 & =\frac{2-10\eta^{2}+10\eta-2\eta^{3}}{K\left(1-3\eta^{2}+2\eta\right)^{2}}\sigma_{\eta}^{2}C(\{\bm{z}_{m},b\}).
\end{align*}
This concludes the result in Theorem \ref{thm:multiple sources}.

\bibliographystyle{IEEEtran}
\bibliography{IEEEabrv,StringDefinitions,JCgroup,ChenBibCV}

\end{document}